\documentclass[a4paper]{arxiv}

\usepackage{hyperref}
\usepackage{dsfont}
\usepackage{color}

\usepackage[latin1]{inputenc}

\usepackage{xspace}
\usepackage{latexsym}
\usepackage{amsmath,amstext,amsxtra,amsgen,amsbsy,amsopn,amscd}
\usepackage{latexsym}
\usepackage{mathrsfs}
\usepackage{dsfont}
\usepackage[active]{srcltx}

\numberwithin{equation}{section}

\newcommand{\ii}{\infty}

\newcommand{\R}{\mathbb{R}}
\newcommand{\N}{\mathbb{N}}
\newcommand{\NN}{\mathcal{N}}
\newcommand{\C}{\mathbb{C}}

\newcommand{\E}{\mathcal{E}}
\newcommand{\cE}{\mathcal{E}}
\newcommand{\cC}{\mathcal{C}}
\newcommand{\cF}{\mathcal{F}}
\newcommand{\cL}{\mathscr{L}}

\newcommand{\Q}{\mathcal{Q}}
\newcommand{\cQ}{\mathcal{Q}}

\newcommand{\Ll}{\mathscr{L}}

\newcommand{\gS}{\mathfrak{S}}
\newcommand{\coul}{\mathcal{C}}

\newcommand{\supp}{\mathrm{supp}}
\newcommand\pscal[1]{{\ensuremath{\left\langle #1 \right\rangle}}}
\newcommand{\norm}[1]{ \left| \! \left| #1 \right| \! \right| }
\newcommand{\wto}{\rightharpoonup}
\renewcommand{\epsilon}{\varepsilon}

\def\XXint#1#2#3{{\setbox0=\hbox{$#1{#2#3}{\int}$}
     \vcenter{\hbox{$#2#3$}}\kern-.5\wd0}}

\DeclareMathOperator{\Tr}{{\rm Tr}}
\DeclareMathOperator{\Tro}{{\rm Tr}_0}
\newcommand{\tr}{\Tr}
\newcommand{\bral}{\left<}
\newcommand{\brar}{\right|}
\newcommand{\ketl}{\left|}
\newcommand{\ketr}{\right>}

\newcommand{\dem}{\varepsilon_{\rm M}}

\newcommand{\FSm}{\gamma_{\rm per} ^0}
\newcommand{\FSmp}{\left(\gamma_{\rm per} ^0 \right) ^{\perp}}
\newcommand{\FSh}{H_{\rm per} ^0}
\newcommand{\FSl}{\epsilon_{\rm F}}

\newcommand{\FSp}{V_{\rm per} ^0}

\newcommand{\rhoQ}{\rho_Q}

\newcommand{\rhoP}{\rho_{\Psi}}

\newcommand{\crysf}{\mathcal{F}_{\rm crys} }
\newcommand{\cryse}{F_{\rm crys} }

\newcommand{\Qpp}{Q ^{++}}
\newcommand{\Qpm}{Q ^{+-}}
\newcommand{\Qmp}{Q^{-+}}
\newcommand{\Qmm}{Q^{--}}

\newcommand{\one}{{\ensuremath {\mathds {1}} }}
\newcommand{\oneep}{\one_{\left(-\infty, \FSl \right)}}

\newcommand{\Sch}{\mathfrak{S}}

\newcommand{\chin}{\chi_{R_n}}
\newcommand{\etan}{\eta_{R_n}}

\newcommand{\chiR}{\chi_R}
\newcommand{\etaR}{\eta_R}

\newcommand{\togeo}{\rightharpoonup_g}
\newcommand{\FockN}{\mathcal{F} ^{\leq N}}

\newtheorem{asumption}{Asumption}[section]

\numberwithin{equation}{section}

\title{\Large On the binding of polarons in a mean-field quantum crystal}
\runningtitle{\textsc{M. Lewin, N. Rougerie}, Binding of small polarons}

\author{Mathieu LEWIN}
\address{CNRS \& Department of Mathematics (UMR 8088), University of Cergy-Pontoise,\\ 95 000 Cergy-Pontoise,
France\\ Email: \email{mathieu.lewin@math.cnrs.fr}}
 
\author{Nicolas ROUGERIE}
\address{Universit\'e Grenoble 1 \& CNRS,  LPMMC (UMR 5493),\\ B.P. 166, 38 042 Grenoble,
France\\ Email: \email{nicolas.rougerie@grenoble.cnrs.fr}}

\date{August 11, 2012}

\begin{document}

\maketitle

\bigskip

\begin{abstract}
We consider a multi-polaron model obtained by coupling the many-body Schr\"odinger equation for $N$ interacting electrons with the energy functional of a mean-field crystal with a localized defect, obtaining a highly non linear many-body problem. The physical picture is that the electrons constitute a charge defect in an otherwise perfect periodic crystal. A remarkable feature of such a system is the possibility to form a bound state of electrons via their interaction with the polarizable background. We prove first that a single polaron always binds, i.e. the energy functional has a minimizer for $N=1$. Then we discuss the case of multi-polarons containing $N\geq2$ electrons. We show that their existence is guaranteed when certain quantized binding inequalities of HVZ type are satisfied.\\
\medskip

\noindent{\scriptsize\copyright~2012 by the authors. This paper may be reproduced, in its entirety, for non-commercial~purposes.}
\end{abstract}

\tableofcontents

\bigskip

\section{Introduction}
A quantum electron in a crystal may form a bound state by using the deformation of the medium which is generated by its own charge~\cite{AleDev-09}. The resulting quasi-particle, composed of the electron and its polarization cloud, is called a \emph{polaron} in the physics literature. Likewise, a \emph{multi-polaron} or \emph{$N$-polaron} is the system formed by the interaction of $N$ electrons with a  crystal. 

That a polaron can be in a bound state is a rather simple physical mechanism. When the (negatively charged) electron is added to the medium, it locally repels (respectively attracts) the other electrons (respectively the positively charged nuclei) of the crystal. A local deformation is thus generated in the crystal, and it is itself felt by the added particle. In other words the additional electron  carries a ``polarization cloud'' with it. It is therefore often useful to think of the polaron as a \emph{dressed particle}, that is a single (composite) particle with new physical properties: effective mass, effective charge, etc.
For an $N$-polaron the situation is a bit more involved. Since the effective polarization has to overcome the natural Coulomb repulsion between the particles, bound states do not always exist.

The question of what model to use to describe the polaron is an important and non trivial one. In the Born-Oppenheimer approximation, a quantum crystal is a very complicated object, made of infinitely many classical nuclei and delocalized electrons. The accurate description of such a system is a very delicate issue and, for this reason, simple effective models are often considered. They should remain  mathematically tractable while still capturing as much of the physics of the system as possible.

A famous example is the model of Fr\"ohlich~\cite{Frohlich-37,Frohlich-52} dating back from 1937, in which the crystal is described as an homogeneous quantized polarization field with which the electrons interact. In the limit of strong coupling between the electrons and the field, the model reduces to Pekar's theory~\cite{Pekar-54,Pekar-63,PekTom-51,LieTho-97,MiySpo-07}. There the crystal is a classical continuous polarizable model, leading to an effective attractive Coulomb interaction in the energy functional of the theory:
\begin{equation}
\cE^{\rm P}_{\varepsilon_{\rm M}}[\psi]=\frac{1}{2}\int_{\R^3}|\nabla\psi(x)|^2\,dx+\frac{(\varepsilon_{\rm M})^{-1}-1}{2}\int_{\R^3}\int_{\R^3}\frac{|\psi(x)|^2|\psi(y)|^2}{|x-y|}dx\,dy.
\label{eq:Pekar-intro}
\end{equation}
Here $\psi$ is the wave-function of the electron, $\dem>1$ is the static dielectric constant of the crystal and we work in atomic units. The variational equation corresponding to \eqref{eq:Pekar-intro} is sometimes called the \emph{Schrödinger-Newton} or \emph{Choquard} equation.

It is the attractive Coulomb term in \eqref{eq:Pekar-intro} that leads to the existence of bound states of electrons, i.e. minimizers (or ground states) of the energy functional. Whereas the energy functional for electrons in vacuum has no minimizer, Lieb~\cite{Lieb-77} proved the existence and uniqueness (up to translations) of a ground state for Pekar's functional~\eqref{eq:Pekar-intro}.  

The same nonlinear attractive term is obtained in Pekar's model for the $N$-polaron. Then, as we have already mentioned, depending on the strength of the attractive Coulomb term as compared to the natural repulsion between the electrons, one can get binding or not. It is an important issue to determine in which parameter range binding occurs~\cite{GriMol-10,FraLieSeiTho-10,FraLieSeiTho-11,Lewin-11}.

\medskip

The approximations made in the construction of Fr\"ohlich's and Pekar's models reduce their applicability to situations where the $N$-polaron is spread over a region of space much larger than the characteristic size of the underlying crystal. One then speaks of \emph{large polarons}. In \cite{LewRou-11} we have introduced a new polaron model by coupling the energy functional for electrons in the vacuum to a microscopic model of quantum crystals with defects introduced in \cite{CanDelLew-08a,CanDelLew-08b}. Unlike in Fr\"ohlich and Pekar theories we take the crystal explicitly into account and make no assumption on the size of the electron. Our approach thus qualifies for the description of both small and large polarons. The model takes the following form (for one electron): 
\begin{equation}
\cE_{\rm eff}[\psi]=\frac1{2}\int_{\R^3}|\nabla\psi(x)|^2\,dx+\int_{\R^3}V^0_{\rm per}(x)|\psi(x)|^2\,dx+F_{\rm crys}\big[|\psi|^2\big]
\label{eq:model-intro1}.
\end{equation}
Here $\FSp$ is the (periodic) electric potential generated by the unperturbed crystal, which is felt by any particle added to the system. The nonlinear effective energy $F_{\rm crys}$ represents the interaction energy between the electrons and the crystal. It is defined using a reduced Hartree-Fock theory for the response of the electrons of the crystal to a charge defect. The state of the Fermi sea of the perturbed crystal is given by a one-body density matrix $\gamma$, that is a non-negative self-adjoint operator on $L^2 (\R ^3)$.  As in \cite{CanDelLew-08a,CanDelLew-08b}, we write  
\begin{equation}\label{eq:perturbed crystal}
\gamma = \FSm + Q
\end{equation}
where $\FSm$ is the density matrix of the periodic unperturbed crystal and $Q$ is the local deformation induced by the charge defect $|\psi|^2$.  The effective energy $\cryse$ then takes the form
\begin{equation}
F_{\rm crys}\big[|\psi|^2\big]=\inf_{-\FSm \leq Q \leq 1-\FSm} \left(\int_{\R^3}\int_{\R^3}\frac{\rhoQ (x) |\psi(y)|^2}{|x-y|}\,dx\,dy+\cF_{\rm crys}[Q]\right).
\label{eq:model-intro2}
\end{equation}
Three main ingredients enter in \eqref{eq:model-intro2}:
\begin{itemize}
\item Electrons are fermions and must thus satisfy the Pauli exclusion principle, which gives in the formalism of density matrices the constraint $0 \leq \gamma \leq 1$ as operators. This justifies the constraint on admissible perturbations $Q$ imposed in \eqref{eq:model-intro2}.
\item The electrons forming the polaron interact with the perturbation they induce in the Fermi sea. This is taken into account by the first term in \eqref{eq:model-intro2} where $\rhoQ$ is the charge density associated with $Q$, given formally by $\rhoQ (x) = Q(x,x)$ (we use the same notation for the operator $Q$ and its kernel).
\item Generating a deformation of the Fermi sea has an energetic cost, represented by the functional $\crysf$ in \eqref{eq:model-intro2}. The  somewhat complicated definition of this functional will be recalled below. It was derived in \cite{CanDelLew-08a}.   
\end{itemize}
More details on how we arrived at the form above can be found in the introduction of \cite{LewRou-11}. Let us mention that this model only takes into account the displacement of the electrons of the crystal and neglects that of the nuclei. This is arguably an important restriction, but our model already captures important physical properties of the polaron, and on the other hand this is all we can treat from a mathematical point of view at present.

In this paper we will show that a (single) polaron described by the energy functional~\eqref{eq:model-intro2} always binds. The case of $N$-polarons is more sophisticated, as now the effective attraction resulting from the polarization of the crystal has to overcome the electronic repulsion. The energy functional corresponding to \eqref{eq:model-intro1} in the case of the $N$-polaron is given by 
\begin{multline}\label{eq:model-introN}
\cE_{\rm eff}[\Psi] = \int_{\R^{3N}}\left(\frac12\sum_{j=1}^N \left|\nabla_{x_j}\Psi(x_1,...,x_N)\right|^2+\sum_{1\leq k<\ell\leq N}\frac{|\Psi(x_1,...,x_N)|^2}{|x_k-x_\ell|}\right)dx_1\cdots dx_N  \\
+\int_{\R ^3} \FSp \rhoP + \cryse [\rhoP]
\end{multline}
where $\rho_\Psi$ is the usual density of charge associated with the many-body wave function $\Psi$ whose definition is recalled in \eqref{eq:density} below.

In fact, our model \eqref{eq:model-intro1} is closely related to Pekar's functional. We proved in \cite{LewRou-11} that Pekar's theory can be recovered from \eqref{eq:model-intro1} in a macroscopic limit where the characteristic size of the underlying crystal goes to $0$. Let us emphasize that our macroscopic limit is completely different from the strong coupling limit of the Fr\"ohlich polaron, which leads to the same Pekar energy \cite{LieTho-97}. It is also associated with a somewhat different physics. In the Fröhlich model, the crystal is polar and it is the deformation of the lattice that binds the polaron, whereas in our case the crystal is initially non polar and only the delocalized Fermi sea gets polarized. The nuclear lattice is not allowed to be deformed in our simplified model.

In addition to clarifying the physics entering the Pekar model, the macroscopic limit argument also gives some interesting insight on the model \eqref{eq:model-intro1}, in particular, regarding the question of the existence of binding. Indeed, it is known \cite{Lewin-11} that Pekar's functional has a ground state in some range of parameters. We deduced in \cite{LewRou-11} that sequences of approximate minimizers for \eqref{eq:model-introN} converge in the macroscopic limit to a ground state of the Pekar functional, thus showing that our model at least accounts for the binding of large polarons in this regime. In this paper, we want to derive conditions ensuring that there is binding in the case of small polarons where the macroscopic limit argument and the link to Pekar's theory are irrelevant.

\medskip

Quite generally, for many-body quantum systems, the existence of bound states of $N$ particles depends on the validity of so-called \emph{binding inequalities}. If $E(N)$ denotes the infimum energy of some physical system containing $N$ particles, a ground state containing $N$ particles exists when 
\begin{equation}\label{eq:binding abstract}
E(N) < \min_{k = 1 ... N} E(N-k) + E ^{\infty} (k)
\end{equation}
where $E ^{\infty}(N)$ denotes the energy of the same $N$ particle system, but with all particles `sent to infinity'. For example, for atoms or molecules comprising $N$ electrons, $E(N)$ includes the contribution of the electric potential generated by the fixed nuclei, while $E^{\infty}(N)$ does not. Particles `at infinity' no longer see the attraction of the nuclei. Note the formal similarity between inequalities \eqref{eq:binding abstract} and those appearing in Lions' concentration compactness principle \cite{Lions-84,Lions-84b}, an important mathematical tool used in nonlinear analysis. The major difference is that the former are quantized and thus more difficult to relate to one another. See \cite{Lewin-11} for a more precise discussion of this connection.

It is not difficult to discuss on physical grounds why inequalities \eqref{eq:binding abstract} are sufficient for the existence of bound states. Indeed,~\eqref{eq:binding abstract} says that sending particles to infinity is not favorable from an energetic point of view. In mathematical terms, the inequalities~\eqref{eq:binding abstract} avoid the \emph{lack of compactness at infinity} of minimizing sequences. The existence of a ground state then follows from the \emph{local compactness} of the model under consideration.
Nevertheless, the mathematical proof that inequalities of the type \eqref{eq:binding abstract} are sufficient for the existence of bound states of $N$-particles is highly non-trivial because the problems $E(N)$, $E(N-k)$ and $E^\ii(k)$ are set in different Hilbert spaces. In the case of atoms and molecules, the fact that inequalities of the form \eqref{eq:binding abstract} imply the existence of bound states is the content of the famous HVZ theorem, first proved independently in \cite{Hun-66,VanWinter-64,Zhislin-60}. 

In this paper we prove an HVZ-type theorem for our polaron functional \eqref{eq:model-introN} when $N\ge2$. We have to face two difficulties. First the functional is invariant under the action of arbitrarily large translations (those leaving invariant the periodic lattice of the crystal), so the energy functional does not change when particles are sent to infinity. The correct binding inequalities therefore take the form
\begin{equation}\label{eq:binding abstract2}
E(N) < \min_{k = 1 ... N-1} E(N-k) + E (k).
\end{equation}
Second, the energy contains the highly nonlinear term $\cryse [\rhoP]$. 
We are thus faced with the combination of the difficulties associated with many-body theory and those inherent to nonlinear problems. A general technique has been introduced in \cite{Lewin-11} to tackle these questions. Our purpose in this paper is to explain how one can deal with the model \eqref{eq:model-introN} using the method of~\cite{Lewin-11}. Our main task will be to control the behavior of the (highly nonlinear) effective polarization energy $F_{\rm crys}$.

In this paper we are not able to show the validity of the binding inequalities~\eqref{eq:binding abstract2} in full generality for $N\geq2$, as this will highly depend on the microscopic structure of the crystal and of the number $N$ of electrons. It should be noticed that, when it occurs, binding is presumably only due to a correlation effect, since in general the effective attraction is weaker than the Coulomb repulsion (see Lemma 1.1 in \cite{LewRou-11}). In the Pekar case, this was explained using Van Der Waals forces in Section 5.3 of~\cite{Lewin-11}.

\bigskip

\noindent\textbf{Acknowledgement.}
The research leading to these results has received funding from the European Research Council under the European Community's Seventh Framework Programme (FP7/2007--2013 Grant Agreement MNIQS no. 258023). We thank Salma Lahbabi for having mentioned to us a mistake in a former version.

\section{Statement of the main results}

\subsection{The mean-field crystal}
We begin by recalling the precise definition of the crystal functional entering in \eqref{eq:model-intro1}. More details can be found in~\cite{CanDelLew-08a,CanDelLew-08b,LewRou-11}.

We fix an $\cL$-periodic density of charge $\mu^0_{\rm per}$ for the classical nuclei of the crystal, with $\cL$ a discrete subgroup of $\R^3$. It is enough for our purpose to assume that $\mu^0_{\rm per}$ is a locally-finite non-negative measure, such that
$\int_{\Gamma}\mu_{\rm per}^0=Z\in\N$, where $\Gamma=\R^3/\cL$ is the unit cell.

In reduced Hartree-Fock theory, the state of the electrons in the crystal is described by a \emph{one-particle density matrix}, which is a self-adjoint operator $\gamma:L^2(\R^3)\to L^2(\R^3)$ such that $0\leq\gamma\leq1$ (in the sense of operators). When no external field is applied to the system, the electrons arrange in a periodic configuration $\gamma=\gamma^0_{\rm per}$, which is a solution of the \emph{reduced Hartree-Fock equations}\footnote{Sometimes called \emph{Hartree equations} in the physics literature.}
\begin{equation}
\begin{cases}
\gamma^0_{\rm per}=\oneep\big(-\Delta/2+V^0_{\rm per}\big),\\[0.2cm]
-\Delta V^0_{\rm per}=4\pi\big(\rho_{\gamma^0_{\rm per}}-\mu^0_{\rm per}\big),\\[0.2cm]
\displaystyle\int_{\Gamma}\rho _{\gamma^0_{\rm per}}=\int_\Gamma\mu^0_{\rm per} = Z.
\end{cases}
\label{eq:rHF-periodic}
\end{equation}
Here $\rho_{A}$ denotes the density of the operator $A$ which is formally given by 
$\rho_A(x)=A(x,x)$ when $A$ is locally trace-class. Also, $\oneep(H)$ denotes the spectral projector of $H$ onto the interval $(-\ii,\FSl)$. The real number $\FSl$ in~\eqref{eq:rHF-periodic} is called the \emph{Fermi level}. It is also the Lagrange multiplier used to impose the constraint that the system must be locally neutral (third equation in~\eqref{eq:rHF-periodic}). 
The unique solution to the self-consistent equation~\eqref{eq:rHF-periodic} is found by minimizing the so-called reduced Hartree-Fock energy functional~\cite{CatBriLio-01,CanDelLew-08a}.

We are working in atomic units with the mass $m$ and the charge $e$ of the electrons of the crystal set to $m=e=1$. Also we neglect their spin for simplicity (reinserting the spin in our model is straightforward).

By Bloch-Floquet theory (see Chapter XIII, Section 16 of ~\cite{ReeSim4}), the spectrum of the $\cL$-periodic Schrödinger operator 
$$H^0_{\rm per}=-\frac12 \Delta+V^0_{\rm per}(x)$$
is composed of bands. When there is a gap between the $Z$th and the $(Z+1)$st bands, the crystal is an insulator and $\FSl$ can be any arbitrary number in the gap. As in~\cite{CanDelLew-08a}, in the whole paper we will assume that the host crystal is an insulator.

\begin{asumption}[\textbf{The host crystal is an insulator}]\label{asum:insulator}\ \\
\it
The periodic Schrödinger operator $H^0_{\rm per}$ has a gap between its $Z$th and $(Z+1)$st bands, and we fix any chemical potential $\FSl$ in the corresponding gap.
\end{asumption}

When the quantum crystal is submitted to an external field, the Fermi sea polarizes.  The method used in~\cite{CanDelLew-08a} to define the energetic cost of such a polarization relies on the following idea. The energetic cost to move the electrons from $\gamma^0_{\rm per}$ to $\gamma$ is defined as the (formal) difference between the (infinite) reduced Hartree-Fock energies of $\gamma$ and of $\gamma^0_{\rm per}$. Denoting by
\begin{equation}
 \label{Coulomb}
D(f,g):= \iint_{\R^3 \times \R ^3} \frac{f (x)g (y)}{|x-y|} dxdy = 4\pi\int _{\R^3} \frac{\overline{\hat{f} (k)} \hat{g} (k)}{|k| ^2} \, dk
\end{equation}
the Coulomb interaction (where $\hat{f}$ denotes the Fourier transform of $f$), one arrives at the functional
\begin{equation}
\cF_{\rm crys}[Q]:=\tr_0 \left((H^0_{\rm per}-\FSl)Q\right)+\frac12 D(\rho_Q,\rho_Q)
\label{eq:def_energy_crystal} 
\end{equation}
where $\Tro$ denotes a generalized trace, see \eqref{eq:decomp Q} and \eqref{eq:gen kin ener} below. For convenience we also denote 
\begin{equation}
\cF_{\rm crys}[\rho,Q]:=\tr_0 \left((H^0_{\rm per}-\FSl)Q\right)+\frac12 D(\rho_Q,\rho_Q) + D(\rho,\rho_Q).
\label{eq:def_energy_crystal 2} 
\end{equation}
The functional setting in which the terms of these equations make sense is defined as follows. Any operator $Q$ satisfying the constraint 
\begin{equation}\label{eq:contrainte cristal}
- \FSm \leq Q \leq 1 -\FSm
\end{equation}
is decomposed as
\begin{equation}\label{eq:decomp Q}
Q = \Qmm + \Qmp +\Qpp + \Qpm 
\end{equation}
where $\Qmm = \FSm Q \FSm$, $\Qmp = \FSm Q \left(1-\FSm \right)$, and so on. It is proved in \cite{CanDelLew-08a} that for $Q$ satisfying \eqref{eq:contrainte cristal} and $\nu\in L^1(\R^3)\cap L^2(\R^3)$, $\crysf [\nu,Q]$ is finite if and only if $Q$ is in the function space 
\begin{equation}\label{eq:crys space}
\Q = \left\lbrace Q \in \Sch ^2  \Big| Q=Q^*,\: |\nabla |Q \in \Sch ^2, \: \Qpp, \Qmm \in \Sch ^1, \: |\nabla | \Qpp |\nabla |,\: |\nabla | \Qmm |\nabla |  \in \Sch ^1\right\rbrace
\end{equation}
that we equip with its natural norm
\begin{equation}
\Vert Q \Vert_{\Q} = \Vert Q \Vert_{\Sch ^2} + \Vert  \Qpp \Vert_{\Sch ^1 }+ \Vert \Qmm \Vert_{\Sch^1 } + \Vert  |\nabla| Q \Vert_{\Sch ^2 } +\Vert |\nabla| \Qpp |\nabla| \Vert_{\Sch ^1} + \Vert |\nabla| \Qmm |\nabla| \Vert_{\Sch ^1}.
\label{eq:def_norm_Q} 
\end{equation}
The symbols $\Sch^1$ and $\Sch  ^2$ denote the Schatten classes of trace-class and Hilbert-Schmidt operators on $L^2 (\R ^3)$ respectively (see~\cite{Simon-79} and~\cite{ReeSim1}, Chapter 6, Section 6). For operators in $\Q$, the kinetic energy in~\eqref{eq:def_energy_crystal} is defined as 
\begin{equation}\label{eq:gen kin ener}
\tr_0 (H^0_{\rm per}-\FSl)Q=\tr\left(|H^0_{\rm per}-\FSl|^{1/2}\big(Q^{++}-Q^{--}\big)|H^0_{\rm per}-\FSl|^{1/2}\right),
\end{equation}
see~\cite{CanDelLew-08a}. More generally, one can define the \emph{generalized trace} as 
\begin{equation}\label{eq:gen trace}
\Tro Q = \Tr \Qpp + \Tr \Qmm
\end{equation}
when $\Qpp$ and $\Qmm$ are trace-class. Note that $\Tro$ differs from the usual trace $\Tr$, the operators in $\Q$ not being trace-class in general. They nevertheless have an unambiguously defined density $\rho_Q\in L^1_{\rm loc}(\R^3)$ 
(see \cite{CanDelLew-08a}, Proposition 1). It belongs to $L^2(\R^3)$ and to the Coulomb space
\begin{equation}\label{eq:coul space}
\coul = \left\lbrace \rho \:\Big| \: D(\rho,\rho ) ^{1/2} <\infty \right\rbrace
\end{equation}
and it holds by definition
\begin{equation}\label{eq:defi rhoQ}
\Tro (V Q ) = \int_{\R^3} V \rhoQ
\end{equation}
for any $V \in \coul'$.

Having defined in \eqref{eq:def_energy_crystal} the total energetic cost to go from $\FSm$ to $\FSm+Q$, we can give a sense to the energetic response of the crystal to an external density $\nu$. The state of the Fermi sea is obtained by solving the following minimization problem
\begin{equation}
\boxed{F_{\rm crys}[\nu]=\inf_{-\gamma^0_{\rm per}\leq Q\leq 1-\gamma^0_{\rm per}}\Big(D(\nu,\rho_Q)+\cF_{\rm crys}[Q]\Big).}
\label{eq:def_F_crys}
\end{equation}
As shown in~\cite{CanDelLew-08a}, for any $\nu\in L^1(\R^3)\cap L^{2}(\R^3)$, this minimization problem has at least one solution in  $\Q$. The corresponding density $\rho_Q$ is in $L^2(\R^3)$ but in general it has long range oscillations which are not integrable at infinity~\cite{CanLew-10}. 

\subsection{The small polaron}
To our crystal we now add $N$ quantum particles, which are by assumption distinguishable from those of the crystal. In reality they are electrons having the same mass $m=1$ as those of the crystal, but we want to keep $m$ arbitrary to emphasize that in our model the additional particles behave differently from those of the crystal. This will also allow us to compare with the results we have obtained in~\cite{LewRou-11}.

The total energy of the system now includes the term $F_{\rm crys}[\nu]$ with $\nu=|\psi|^2$ (polaron) or $\nu=\rho_\Psi$ ($N$-polaron). For the single polaron, the energy is given by
\begin{equation}\label{eq:funct1}
\cE [\psi] = \int_{\R^{3}}\left(\frac{1}{2m}|\nabla\psi(x)|^2+V^0_{\rm per}(x)|\psi(x)|^2\right)\,dx + \cryse \big[|\psi|^2\big].
\end{equation}
For the $N$-polaron with $N\geq2$ it reads
\begin{multline}\label{eq:functN}
\cE [\Psi] = \int_{\R^{3N}}\left(\frac1{2m}\sum_{j=1}^N \left|\nabla_{x_j}\Psi(x_1,...,x_N)\right|^2+\sum_{1\leq k<\ell\leq N}\frac{|\Psi(x_1,...,x_N)|^2}{|x_k-x_\ell|}\right)dx_1\cdots dx_N  \\
+\int_{\R ^3} \FSp(x) \rhoP(x)\,dx + \cryse [\rhoP].
\end{multline}
As we think that there is no possible confusion, we do not emphasize the particle number $N$ in our notation of the energy $\cE$.
The density $\rhoP$ is defined as 
\begin{equation}\label{eq:density}
\rhoP (x) = N\int_{\R ^{3(N-1)}} \left| \Psi (x,x_2,\ldots,x_N)\right| ^2 dx_2\ldots dx_N.
\end{equation}
The corresponding ground state energies read
\begin{equation}\label{eq:energy1}
E(1): = \inf \left\lbrace \E [\psi],\: \psi \in H ^1 (\R ^{3}),\ \int_{\R ^{3} } |\psi| ^2 = 1 \right\rbrace 
\end{equation}
and
\begin{equation}\label{eq:energyN}
E(N): = \inf \left\lbrace \E [\Psi],\: \Psi \in H ^1 (\R ^{3N}), \: \Psi \mbox{ fermionic, } \int_{\R ^{3N} } |\Psi| ^2 = 1 \right\rbrace. 
\end{equation}
Here by `fermionic' we mean antisymmetric under particle exchange:
\begin{equation}\label{eq:fermionic}
\Psi(x_1,\ldots,x_i,\ldots,x_j,\ldots,x_N) = -\Psi(x_1,\ldots,x_j,\ldots,x_i,\ldots,x_N) \mbox{ for any } i\neq j
\end{equation}
as is appropriate for electrons. Recall that we have neglected the spin for simplicity.

We now state our main results. In the single polaron case we are able to show the existence of a bound state.

\begin{theorem}[Existence of small polarons]\label{theo:E_1}\ \\
For $N=1$, we have 
\begin{equation}
E(1)<E_{\rm per}:=\inf\sigma\left(-\frac1{2m}\Delta+V^0_{\rm per}\right).
\label{eq:E_1} 
\end{equation}
There are always minimizers for $E(1)$ and all the minimizing sequences converge to a minimizer for $E(1)$ strongly in $H^1(\R^3)$, up to extraction of a subsequence and up to translations.
\end{theorem}

Inequality \eqref{eq:E_1} expresses the fact that binding is energetically favorable : the right-hand side is the energy an electron would have in absence of binding.

In the $N$-polaron case we can give necessary and sufficient conditions for the compactness of minimizing sequences.

\begin{theorem}[HVZ for small $N$-polarons]\label{theo:HVZ}\ \\
For $N\geq2$, the following assertions are equivalent:
\begin{enumerate}
\item One has 
\begin{equation}\label{eq:binding}
 E(N) < E(N-k) + E(k) \mbox{ for all } k=1,\ldots, N-1.
\end{equation}
\item Up to translation and extraction of a subsequence, all the minimizing sequences for $E(N)$ converge to a minimizer for $E(N)$ strongly in $H ^1 (\R ^{3N})$.
\end{enumerate}
\end{theorem}

\begin{remark}
For this result, the fermionic nature of the particles inserted into the crystal is not essential. The same theorem holds if they are replaced by bosons , i.e. the wave function $\Psi$ is supposed to be symmetric under particle exchange.
\end{remark}

As discussed in the introduction, this theorem is rather natural from a physical point of view. It is not expected that the conditions \eqref{eq:binding} hold in general. As in Pekar's theory, one should expect the existence of minimizers to depend on the choice of parameters entering the functional (in our case only the periodic distribution $\mu^0_{\rm per}$ of the nuclei). Testing the validity of these inequalities is a challenging task that would require more knowledge on the properties of the crystal model than we presently have. In particular, the decay at infinity of the minimizers of the crystal model should be investigated.

In~\cite{LewRou-11} we have considered a macroscopic regime where the mass $m$ of the polarons tend to zero. In this limit $m\to0$ the ground state energy $E_m(N)$ converges to Pekar's energy involving the macroscopic dielectric constant $\varepsilon_{\rm M}$ of the crystal defined in~\cite{CanLew-10} (up to a simple oscillatory factor, see~\cite{LewRou-11} for details). It was shown in~\cite{Lewin-11} that the binding inequalities are satisfied in Pekar's theory when $\varepsilon_{\rm M}$ is large enough. We conclude that in this case they will also be satisfied for $m$ small enough and therefore minimizers do exist in this case.

The rest of the paper is devoted to the proof of Theorem \ref{theo:HVZ}. One of us has considered in Section 5 of \cite{Lewin-11} a general class of nonlinear many-body problems of the form
\[
\int_{\R^{3N}}\left(\frac12\sum_{j=1}^N \left|\nabla_{x_j}\Psi(x_1,...,x_N)\right|^2+\sum_{1\leq k<\ell\leq N} |\Psi(x_1,...,x_N)|^2 W(x_k-x_l)\right)dx_1\cdots dx_N + F [\rhoP]
\]
and provided sufficient assumptions on the interaction potential $W$ and the non linearity $F$ under which a HVZ type result similar to 
Theorem \ref{theo:HVZ} holds. The assumptions on $W$ include the Coulomb interaction we are concerned with in this paper but, unfortunately, our crystal functional $\cryse$ does not seem to satisfy all the properties imposed on $F$ in~\cite{Lewin-11}. Also the presence of the periodic potential $\FSp$ adds a new difficulty. Nevertheless the general strategy of \cite{Lewin-11} still applies and our goal in this paper is to explain how to overcome the difficulties associated with $\cryse$.

Section \ref{sec:prop crystal} gathers some important properties of the crystal functional that are to be used in the proofs of Theorems \ref{theo:E_1} and \ref{theo:HVZ}, presented in Sections \ref{sec:proof1} and \ref{sec:proof} respectively.

\section{Properties of the crystal energy}\label{sec:prop crystal}

In this section we roughly speaking prove that $\cryse$ satisfies Assumptions (A1) to (A5) of \cite{Lewin-11}, Section 5. We are only able to prove a little less, but the properties we do prove are sufficient for the proof of Theorem~\ref{theo:HVZ} as we explain in Section \ref{sec:proof}. 

We start in Section~\ref{sec:prop crystal 1} with almost immediate consequences of the definition of $\cryse$, and devote Section~\ref{sec:prop crystal 2} to the more involved fact that our crystal functional satisfies a `decoupling at infinity' property. The proof of this property requires some facts about localization operators that we gather in Section~\ref{sec:localization}.

\subsection{Concavity, subcriticality and translation invariance}\label{sec:prop crystal 1}

The following is the equivalent of Assumptions (A4) and (A5) in \cite{Lewin-11}, Section 5.

\begin{lemma}[Concavity]\label{lem:concavity}\mbox{}\\
$\cryse$ is concave on $\left\lbrace \rho \in\cC,\: \rho \geq 0 \right\rbrace$. Moreover it is strictly concave at the origin:
\begin{equation}\label{eq:strict concavity}
\cryse[t\rho] > t \cryse [\rho] 
\end{equation}
for all $\rho \in \cC \setminus \left\lbrace 0 \right\rbrace$, $\rho \geq 0$ and all $0<t<1$.
\end{lemma}

\begin{proof}
The functional $\crysf[\rho,Q]$ defined in \eqref{eq:def_energy_crystal 2} is linear in $\rho$. 
As by definition 
\[
 \cryse [ \rho ] = \inf \left\lbrace \crysf [ \rho ,Q ], -\FSm \leq Q \leq 1-\FSm \right\rbrace,
\]
it is clearly a concave functional of $\rho$.\
As for the strict concavity we note that
\[
 \crysf[t\rho,Q] = \Tro \left( \left(\FSh - \FSl\right) Q \right) + \frac{1}{2} D(\rho_Q,\rho_Q) + t D (\rho,\rhoQ) > t \crysf[\rho,Q]\geq t \cryse[\rho]
\]
for all $0<t<1$ by positivity of the kinetic and Coulomb energies. Taking for $Q$ the minimizer corresponding to $t\rho$ which is known to exist by \cite{CanDelLew-08a,CanLew-10} proves \eqref{eq:strict concavity}.
\end{proof}

The next lemma will be useful to prove that minimizing sequences for our polaron model are bounded in $H ^1 (\R ^{3N})$. It is the equivalent of Assumption (A3) in \cite{Lewin-11}, Section 5.

\begin{lemma}[Subcriticality]\label{lem:subcritical}\mbox{}\\
The functional $\cryse$ is locally uniformly continuous on $L ^{6/5}$. More precisely, we have
\begin{equation}
\big| \cryse[\rho] - \cryse[\rho '] \big| \leq C \left\Vert \rho- \rho ' \right\Vert_{L ^{6/5}}^2
\label{eq:loc_Lipschitz}
\end{equation}
for a universal constant $C>0$. Moreover, for every $\epsilon>0$, we have
\begin{equation}\label{eq:subcritical}
0>\cryse [|\varphi| ^2] \geq - \epsilon \int_{\R ^3} |\nabla \varphi| ^2-\frac{C}{\epsilon}\left(\int_{\R^3}|\varphi|^2\right)^3 
\end{equation}
for all $\varphi \in H ^1 (\R ^3)$.
\end{lemma}

\begin{proof}
For any $\rho \in L ^{6/5}$ and any $Q \in Q$ we can complete the square in the electrostatic terms of $\crysf[\rho,Q]$ and obtain
\[
\crysf[\rho,Q] = \Tro \left( \left(\FSh - \FSl\right) Q \right) + \frac{1}{2} D(\rho_Q+ \rho,\rho_Q +\rho) -  \frac{1}{2} D (\rho,\rho) \geq - \frac{1}{2} D(\rho,\rho). 
\]
Taking the infimum with respect to $Q$ and applying this with $\rho = |\varphi| ^2$ immediately yields 
$$\cryse[|\varphi|^2]\geq -\frac12 D(|\varphi|^2,|\varphi|^2)\geq -C\norm{\varphi}_{L^{12/5}}^4$$
by the Hardy-Littlewood-Sobolev inequality (\cite{LieLos-01}, Theorem 4.3). Using now the Sobolev and Hölder inequalities we get as stated 
$$\norm{\varphi}_{L^{12/5}}^4\leq \norm{\varphi}_{L^6}\norm{\varphi}_{L^2}^3\leq \epsilon\int_{\R^3}|\nabla\varphi|^2+\frac{C}{\epsilon}\left(\int_{\R^3}|\varphi|^2\right)^3.$$
Then, replacing $\rho$ by $\rho-\rho'$ we also have
\[
 \crysf[\rho-\rho',Q] \geq - \frac{1}{2} D(\rho-\rho',\rho-\rho').
\]
Choosing now for $Q$ a minimizer of $\crysf[\rho,Q]$ we deduce 
\[
 \cryse[\rho] - \cryse[\rho '] \geq - \frac{1}{2} D(\rho-\rho',\rho-\rho').
\]
Without loss of generality we can assume that the left-hand side is negative. We conclude that there exists a constant such that 
\[
 \big| \cryse[\rho] - \cryse[\rho '] \big| \leq C \left\Vert \rho- \rho ' \right\Vert_{L ^{6/5}}^2
\]
using the Hardy-Littlewood-Sobolev inequality again.
\end{proof}

Finally, we note that our functional is invariant under the action of the translations of the periodic lattice $\Ll$. Note that in \cite{Lewin-11}, full translation invariance is assumed (see Assumption (A2)). However, what is really used in the proof of the results there is the invariance under the action of arbitrarily large translations.

\begin{lemma}[Translation invariance]\label{lem:translation}\mbox{}\\
For any $\rho \in L ^{6/5}$ and any translation $\vec{\tau}\in\Ll$ of the periodic lattice,
\begin{equation}\label{eq:translation}
\cryse[\rho \left( \cdot+ \vec{\tau}\right)] =  \cryse[\rho].
\end{equation}
\end{lemma}

\begin{proof}
We denote by $Q$ a minimizer of $\crysf[\rho,Q]$. Clearly
$\rhoQ (\cdot+\vec{\tau}) = \rho_{U_{\vec{\tau}}^{*} Q U_{\vec{\tau}}}$
where $U_{\vec\tau}$ is the unitary translation operator acting on $L^2(\R^3)$ and defined by $U_{\vec\tau}f=f(\cdot-\vec\tau)$.
We deduce 
\begin{eqnarray*}
\cryse[\rho (\cdot+\vec{\tau})] \leq \crysf[\rho (\cdot+\vec{\tau}),\vec{\tau} ^{*} Q \vec{\tau}] = \Tro \left( \vec{\tau} \left( \FSh- \FSl \right) \vec{\tau} ^{*} Q \right) + \frac{1}{2} D (\rho,\rhoQ) -D (\rho,\rhoQ)=\cryse[\rho]
\end{eqnarray*}
by translation invariance of the Coulomb interaction and the fact that $\FSh$ commutes with the translations of the lattice $\cL$. Exchanging the roles of $\rho (\cdot+\vec{\tau})$ and $\rho$ and applying the same argument proves that there must be equality.
\end{proof}

\subsection{Some localization properties}\label{sec:localization}

In order to prove that the crystal energy of two distant clusters of mass decouples we will need a localization procedure. Due to the constraint \eqref{eq:contrainte cristal}, it is convenient to use a specific localization method for $Q_n$, as noted first in~\cite{HaiLewSer-09,CanDelLew-08a}. We here provide several new facts about this procedure that will be useful in the next section.  

We introduce a smooth partition of unity $\chi^2 + \eta^2 = 1$ such that $\chi= 1$ on the ball $B(0,1)$ and $\chi=0$ outside of the ball $B(0,2)$. Similarly, $\eta=1$ on $\R^3\setminus B(0,2)$ and $\eta=0$ on $B(0,1)$. We also require that $\nabla\chi$ and $\nabla\eta$ are bounded functions. Then we introduce $\chi_R(x):=\chi(x/R)$ and $\eta_R(x)=\eta(x/R)$.  
We define the two localization operators
\begin{eqnarray}\label{eq:localizations}
X_R &=& \FSm \chi_R \FSm + \FSmp \chi_R \FSmp \nonumber \\
Y_R &=& \FSm \eta_R \FSm + \FSmp \eta_R \FSmp
\end{eqnarray}
that have the virtue of commuting with the spectral projectors $\FSm$ and $\FSmp=1-\FSm$. 
Note that in~\cite{CanDelLew-08a}, the choice $X_R=\sqrt{1-Y_R^2}$ is made. Here we change a bit the strategy and we only have
$$X_R^2+Y_R^2\leq1.$$
The following lemma, whose lengthy proof shall be detailed in the Appendix, says that $X_R ^2 + Y_R ^2 \approx 1$ for large $R$, in a sufficiently strong sense for our practical purposes, see Section \ref{sec:prop crystal 2}.

\begin{lemma}[Properties of the localization operators $X_R$ and $Y_R$]\label{lem:loc_properties}\mbox{}\\
There exists a universal constant $C>0$ such that
\begin{equation}
\norm{X_RQX_R}_\cQ+\norm{Y_RQY_R}_\cQ +\norm{\rho_{X_RQX_R}}_{L^2\cap \cC}+\norm{\rho_{Y_RQY_R}}_{L^2\cap \cC}\leq C\norm{Q}_\cQ,
\label{eq:uniform_cQ}
\end{equation}
\begin{equation}
\Big|\tr_0(H^0_{\rm per}-\epsilon_{\rm F})Q- \tr_0(H^0_{\rm per}-\epsilon_{\rm F})X_RQX_R-\tr_0(H^0_{\rm per}-\epsilon_{\rm F})Y_RQY_R\Big|\leq \frac{C}{R^2}\norm{Q}_\cQ,
\label{eq:localization_kinetic} 
\end{equation}
and
\begin{equation}
\norm{\rho_{Q} - \rho_{X_RQX_R} - \rho_{Y_RQY_R}}_{L^2\cap\cC} \leq\frac{C}{R}\norm{Q}_\cQ
\label{eq:error_localization_rho}
\end{equation}
for all $Q\in\cQ$ and all $R\geq1$.
\end{lemma}

Remark that the IMS formula implies
$$H^0_{\rm per}-\epsilon_{\rm F}=\chi_R\big(H^0_{\rm per}-\epsilon_{\rm F}\big)\chi_R+\eta_R\big(H^0_{\rm per}-\epsilon_{\rm F}\big)\eta_R-\frac12|\nabla\chi_R|^2-\frac12|\nabla\eta_R|^2$$
where the last two error terms can be estimated in the operator norm by $R^{-2}(\norm{\nabla\chi}_{L^\ii}^2+\norm{\nabla\eta}_{L^\ii}^2)/2$. Our bound~\eqref{eq:localization_kinetic} is a similar estimate valid for the modified localization operators $X_R$ and $Y_R$. In the same spirit, remark that 
$$\rho_Q=\chi_R^2\rho_Q+\eta_R^2\rho_Q=\rho_{\chi_RQ\chi_R}+\rho_{\eta_RQ\eta_R}$$
and therefore the estimate~\eqref{eq:error_localization_rho} on the density quantifies the error when the localization operators $X_R$ and $Y_R$ are used in place of $\chi_R$ and $\eta_R$.

As noticed first in~\cite{HaiLewSer-09}, the main advantage of the localization operators $X_R$ and $Y_R$ is that they preserve the constraint \eqref{eq:contrainte cristal}. Simply, using that 
$$X_R\FSm X_R=\FSm\chi_R\FSm\chi_R\FSm\leq \FSm(\chi_R)^2\FSm\leq \FSm$$
and the similar estimate $X_R\FSmp X_R\leq\FSmp$, we see that when $-\FSm\leq Q\leq 1-\FSm$, then
\begin{equation}\label{eq:preser contrainte}
-\FSm\leq-X_R\FSm X_R\leq X_RQX_R\leq X_R\FSmp X_R\leq \FSmp
\end{equation}
and the same is true for $Y_RQY_R$.

This is in fact a particular case of an algebraic property which does not seem to have been noticed before, that we state as Lemma~\ref{lem:abstract_sum} below. It will be very useful when constructing trial states for the crystal functional.

\begin{lemma}[Adding states using localization]\label{lem:abstract_sum}\mbox{}\\
Let $\Pi$ be an orthogonal projector on a Hilbert space $\mathfrak{H}$, and $\chi$, $\eta$ two self-adjoint operators on $\mathfrak{H}$ such that $\chi^2+\eta^2\leq 1$. We introduce the corresponding localization operators
$$X=\Pi\chi\Pi+(1-\Pi)\chi(1-\Pi)\qquad\text{and}\qquad Y=\Pi\eta\Pi+(1-\Pi)\eta(1-\Pi).$$
Let $Q,Q'$ two self-adjoint operators such that $-\Pi\leq Q,Q'\leq 1-\Pi$. Then we have
\begin{equation}
-\Pi\leq XQX+YQ'Y\leq 1-\Pi
\end{equation}
as well.
\end{lemma}
\begin{proof}
Since $X$ and $Y$ are self-adjoint we have 
\[
-X\Pi X-Y\Pi Y\leq XQX+YQ'Y\leq X(1-\Pi)X+Y(1-\Pi)Y.                                            
\]
The lemma follows from the estimate
$$X\Pi X+Y\Pi Y=\Pi\chi\Pi\chi\Pi+\Pi\eta\Pi\eta\Pi\leq \Pi\big(\chi^2+\eta^2\big)\Pi\leq \Pi$$
and the equivalent one with $\Pi$ replaced by $1-\Pi$. In the last line we have used that $\Pi\leq 1$ and $\chi,\eta$ are self-adjoint to get $\Pi \chi \Pi \chi \Pi \leq \Pi \chi ^2 \Pi$ and $\Pi \eta \Pi \eta \Pi \leq \Pi \eta^2 \Pi$. 
\end{proof}

It will be important in the sequel to know that weak convergence of a sequence $(Q_n)$ in $\Q$ implies strong local compactness, that is strong compactness of $(XQ_nX)$ where $X$ is defined similarly as above, starting from a compactly supported function $\chi$.

\begin{lemma}[Strong local compactness for bounded sequences in $\cQ$]\label{lem:compac dens}\mbox{}\\
Let $(Q_n)$ be a bounded sequence in $\Q$ such that $Q_n\wto Q$ weakly in $\cQ$. Then $\chi Q_n\chi\to \chi Q\chi$
strongly in the trace-class $\gS^1$, for every function $\chi\in L^\ii(\R^3)$ of compact support. In particular, $\rho_{Q_n}\to\rho_Q$ weakly in $L^2\cap \cC$ and strongly in $L^1_{\rm loc}$.\\
Writing $X = \FSm \chi \FSm + (1-\FSm) \chi (1-\FSm)$, we also have $X Q_nX \to XQX$ strongly in $\Sch ^1$ and thus $\rho_{XQ_nX} \to \rho_{XQX}$ strongly in $L^1$. 

\end{lemma}

We will use the following local compactness criterion in Schatten classes. Its standard proof is omitted. 

\begin{lemma}[Local compactness in Schatten spaces]\label{lem:Schatten}\mbox{}\\
Let $\Sch ^p$ be the class of compact operators $A$ of some Hilbert space $\mathfrak{H}$ such that $\left( \Tr (|A| ^p)\right) ^{1/p} < +\infty$, with the convention that $\Sch ^{\infty}$ denotes the class of compact operators.
\begin{itemize}
\item If $A_n \rightharpoonup A$ weakly-$\ast$ in $\Sch ^1$ and $K,K' \in \Sch ^{\infty}$ then $K A_n K' \rightarrow KAK'$ strongly in $\Sch ^1$.

\smallskip

\item If $A_n \rightharpoonup A$ weakly in $\Sch ^r$, $K\in \Sch ^{p}$ and $K' \in \Sch ^{q}$ then $K A_n K' \rightarrow KAK'$ strongly in $\Sch ^s$ with 
${1}/{s} = {1}/{p}+{1}/{q}+{1}/{r}$.
\end{itemize}
\end{lemma}

\noindent\emph{Proof of Lemma \ref{lem:compac dens}}
We know from Proposition 1 in~\cite{CanDelLew-08a} that $\rho_{Q_n}\wto\rho_Q$ weakly in $L^2\cap\cC$. Only the strong local convergence is new. We write as usual
\begin{equation}\label{eq:decomp Qn}
 Q_n = Q_n ^{++} + Q_n ^{-+} + Q_n ^{+-} + Q_n ^{--}
\end{equation}
and consider only the first two terms, the other two being dealt with in a similar way. We have
\[
 \chi Q_n ^{++} \chi = \Big\{\chi \left( -\Delta + 1 \right) ^{-1/2}\Big\} \Big\{\left( -\Delta + 1 \right) ^{1/2} Q_n ^{++} \left( -\Delta + 1 \right) ^{1/2}\Big\} \Big\{\left( -\Delta + 1 \right) ^{-1/2} \chi\Big\}.
\]
The operator $ \chi ( -\Delta + 1 ) ^{-1/2}$ is compact and $( -\Delta + 1 ) ^{1/2} Q_n ^{++} ( -\Delta + 1 ) ^{1/2}$ converges towards $( -\Delta + 1 ) ^{1/2} Q ^{++} ( -\Delta + 1 ) ^{1/2}$ weakly-$\ast$ in $\gS^1$ by assumption. By Lemma~\ref{lem:Schatten} we deduce that $\chi Q_n ^{++} \chi\to\chi Q ^{++} \chi$ strongly in $\gS^1$.

We argue similarly for the off diagonal terms, writing this time
\[
 \chi Q_n ^{+-} \chi = \Big\{\chi \left( -\Delta + 1 \right) ^{-1/2}\Big\} \Big\{\left( -\Delta + 1 \right) ^{1/2} Q_n^{+-}\Big\} \Big\{\FSm \chi\Big\}.
\]
Again the operator $\chi (-\Delta + 1 ) ^{-1/2}$ is compact and we can write
$$\FSm \chi=\FSm(H^0_{\rm per}+\mu)\,(H^0_{\rm per}+\mu)^{-1}\,(1-\Delta)\,(1-\Delta)^{-1} \chi.$$
Here $\mu$ is a large enough constant such that $H^0_{\rm per}\geq -\mu/2$. The operator $\FSm(H^0_{\rm per}+\mu)$ is bounded by the functional calculus. Also, $(H^0_{\rm per}+\mu)^{-1}\,(1-\Delta)$ is bounded by Lemma 1 in~\cite{CanDelLew-08a}. Finally, $(1-\Delta)^{-1} \chi\in\gS^2$. 
Thus $\FSm \chi\in\gS^2$. Since $( -\Delta + 1 ) ^{1/2} Q_n^{+-}\wto ( -\Delta + 1 ) ^{1/2} Q^{+-}$ weakly in $\gS^2$ by assumption, we deduce by Lemma~\ref{lem:Schatten} again, that $\chi Q_n ^{+-} \chi\to\chi Q ^{+-} \chi$ strongly in $\gS^1$. 

We have proved that $\chi Q_n\chi\to\chi Q\chi$ strongly in $\gS^1$, but $\rho_{\chi Q_n\chi} = \chi ^2 \rho_{Q_n}$, so we deduce that  $\rho_{Q_n}\to\rho_Q$ strongly in $L^1_{\rm loc}$.

For the second part of the statement we simply write 
\[
XQ_nX_=\FSmp\chi Q_n^{++}\chi\FSmp+\FSm\chi Q_n^{--}\chi\FSm+\FSmp\chi Q_n^{+-}\chi\FSm\\+\FSm\chi Q_n^{-+}\chi\FSmp 
\]
and use the strong convergence of each term shown above.

\hfill\qed

One can prove that if $\rho \in \coul$ then $\chiR ^2 \rho \to \rho$ strongly in $\coul$ when $R\to \infty$. In the same spirit, we have

\begin{lemma}[\textbf{Approximation using localization}]\label{lem:loc approx}\mbox{}\\
With $X_R$ and $Y_R$ defined as above and $Q\in \cQ$,
\begin{equation}\label{eq:loc approx}
X_R Q X_R \to Q,\quad Y_R Q Y_R \to 0 \mbox{ strongly in } \cQ \mbox{ when } R\to \infty.  
\end{equation}
In particular $\rho_{X_RQX_R} \to \rho_Q$ and $\rho_{Y_R Q Y_R} \to 0$ strongly in $L^2 \cap \coul$.
\end{lemma}

\begin{proof}
Using \eqref{eq:uniform_cQ} and an $\epsilon/2$ argument, it suffices to prove this for a finite rank operator $Q$ (such operators are known to be dense in $\cQ$, see Corollary 3 in \cite{CanDelLew-08a}). In this case the statement just follows from the facts that $X_R\to1$ and $Y_R\to0$ strongly, which is a consequence of the convergence $\chiR \to 1$ and $\etaR \to 0$. The convergence of $\rho_{X_RQX_R}$ and $\rho_{Y_R Q Y_R}$ follows by continuity of the map $Q\in\cQ \mapsto \rhoQ \in L^2\cap \coul$.
\end{proof}


\subsection{Decoupling at infinity}\label{sec:prop crystal 2}

Here we provide the most crucial ingredient of the proof of Theorem~\ref{theo:HVZ}, namely the fact that the crystal energy of the sum of two distant pieces of mass is almost the sum of the energies of these pieces. This is the content of the following proposition, which is the equivalent of assumption (A3) in \cite{Lewin-11}. Note however that we prove much less than what is stated there. Fortunately, the proof of Theorem 25 in \cite{Lewin-11} does not actually require such a strong assumption as (A3), as we will show in Section~\ref{sec:proof} below.

\begin{proposition}[Decoupling at infinity]\label{pro:decouple}\ \\
Let $(\rho_n)$ be a bounded sequence in the Coulomb space $\cC$ such that $\rho_n\wto\rho$ weakly. 
Then
\begin{equation}\label{eq:decouple}
\lim_{n\to \infty} \bigg( \cryse[\rho_n] -  \cryse[\rho] - \cryse[\rho_n-\rho] \bigg) = 0.
\end{equation}
\end{proposition} 

In the above, one should think of $\rho_n$ as being constituted of two clusters of mass, $\rho$ and $\rho_n-\rho$, whose ``supports'' are infinitely far away in the limit $n\to \infty$. This is mathematically materialized by the weak convergence to $0$ of $\rho_n-\rho$. The proposition then says that the total energy is the sum of the energy of the pieces, up to a small error. Proving~\eqref{eq:decouple} is a difficult task because of the long range behavior of the response of the crystal : it is known~\cite{CanLew-10} that the polarization $\rho_Q$ of the Fermi sea has long range oscillations that are not integrable at infinity. The oscillations generated by $\rho$ are seen by $\rho_n-\rho$ (and conversely) but, fortunately, they contribute a small amount to the total energy, which is controlled by the Coulomb norm and not the $L^1$ norm.

Assumption (A3) in \cite{Lewin-11} is a little different from \eqref{eq:decouple}. There it was assumed that $\rho_n=\rho_n^1+\rho_n^2$ where $\rho_n^1$ and $\rho_n^2$ are bounded in $L^{6/5}$ and that the distance between their supports goes to infinity, with no assumption on the size of these supports. In Proposition \ref{pro:decouple} it is implicit that one of the two clusters of mass has a support of bounded size and is approximated by its weak limit $\rho$. This additional assumption is harmless for our purpose because we are dealing with a locally compact problem. 

In the course of the proof of Proposition \ref{pro:decouple} we will establish the following, which we believe is of independent interest. It gives the weak continuity of the (a priori multi-valued) map $\rho_n\mapsto Q_n={\rm argmin}\;\crysf[\rho_n,\cdot]$.

\begin{corollary}[A weak continuity result for $\crysf$]\label{cor:continu_Fcrys}\ \\
Let $(\rho_n)$ be a bounded sequence in the Coulomb space $\cC$ such that $\rho_n\wto \rho$ weakly and $Q_n$ be any minimizer of $\crysf[\rho_n,\cdot]$. Then, up to extraction of a subsequence, $Q_n \wto Q$ weakly in $\Q$ where $Q$ minimizes $\crysf[\rho,\cdot]$.
\end{corollary} 

We now present the 

\medskip

\noindent\textbf{Proof of Proposition~\ref{pro:decouple}.}
We begin with the difficult part, that is the proof of the lower bound corresponding to \eqref{eq:decouple}.

\smallskip

\noindent\textsl{Step 1: Lower bound}. We denote by $Q_n$ a minimizer for $Q\mapsto \crysf[\rho_n,Q]$. Corollary 2 in \cite{CanDelLew-08a} states that the energy functional $\crysf[\rho_n,Q]$ controls the norm $\left\Vert Q\right\Vert_{\Q}$: 
\[
 0 \geq \crysf [\rho_n,Q_n] \geq C \left\Vert Q_n \right\Vert_{\Q} -\frac{1}{2} D(\rho_n,\rho_n).
\]
The upper bound is obtained by taking a trial state $Q\equiv 0$. Using that $\rho_n$ is bounded in $\cC$, hence that $D(\rho_n,\rho_n)$ is bounded, we deduce that the sequence $(Q_n)$ is bounded in $\Q$. Up to extraction of a subsequence, we can assume that $Q_n\wto Q$ and, by Lemma~\ref{lem:compac dens}, that $\rho_{Q_n}\to\rho_Q$ weakly in $L^2\cap\coul$ and strongly in $L^1_{\rm loc}$.

We now consider localization operators $X_R$ and $Y_R$ as described in the preceding section and use them to write the energy as the sum of the energy of $X_RQ_nX_R$ and that of $Y_RQ_nY_R$, modulo errors terms. In our proof $R$ is fixed and will go to infinity only in the end, after we have taken the limit $n\to\ii$.

Using that $\rho_n$ and $\rho_{Q_n}$ are bounded in $\cC$ and that $Q_n$ is bounded in $\cQ$, we get from the estimate~\eqref{eq:error_localization_rho}
\begin{align*}
D(\rho_{Q_n},\rho_n)&=D(\rho_{Q_n},\rho)+D(\rho_{Q_n},\rho_n-\rho)\\
&=D(\rho_{Q},\rho)+D(\rho_{Q_n}-\rho_Q,\rho)+D(\rho_{X_RQ_nX_R},\rho_n-\rho)+D(\rho_{Y_RQ_nY_R},\rho_n-\rho)+\epsilon_n(R)\\
&=D(\rho_{Q},\rho)+D(\rho_{X_RQ_nX_R},\rho_n-\rho)+D(\rho_{Y_RQ_nY_R},\rho_n-\rho)+o(1)+\epsilon_n(R)
\end{align*}
where we have used that $\rho_{Q_n}\wto\rho_Q$ weakly in $\cC$ and where $\epsilon_n (R)$ denotes a generic quantity satisfying
\begin{equation}\label{eq:epsilon(R)}
 \limsup_{n\to\ii}|\epsilon_n(R)|\leq \frac{C}{R}. 
\end{equation}
Also $o(1)$ goes to $0$ when $n\to \infty$ and $R$ stays fixed. Since $\rho_{X_RQ_nX_R}\to\rho_{X_RQX_R}$ strongly in $L^1(\R^3)$ by Lemma \ref{lem:compac dens}, and it is a bounded sequence in $L^2(\R^3)$ by~\eqref{eq:uniform_cQ}, it must converge strongly in $\cC$ by the Hardy-Littlewood-Sobolev inequality. So we conclude that $D(\rho_{X_RQ_nX_R},\rho_n-\rho)\to0$ as $n\to\ii$, hence that
\begin{equation}
D(\rho_{Q_n},\rho_n)=D(\rho_{Q},\rho)+D(\rho_{Y_RQ_nY_R},\rho_n-\rho)+o(1)+\epsilon_n(R).
\end{equation}
Arguing exactly the same, we can conclude that
\begin{align*}
D(\rho_{Q_n},\rho_{Q_n})&=D(\rho_{X_RQ_nX_R}+\rho_{Y_RQ_nY_R},\rho_{X_RQ_nX_R}+\rho_{Y_RQ_nY_R})+ \epsilon_n(R)\\
&=D(\rho_{X_RQX_R},\rho_{X_RQX_R})+2D(\rho_{X_RQX_R},\rho_{Y_RQY_R})+D(\rho_{Y_RQ_nY_R},\rho_{Y_RQ_nY_R})+o(1)+\epsilon_n(R).
\end{align*}

If we use these estimates on the electrostatic terms and~\eqref{eq:localization_kinetic} to deal with the kinetic energy, we arrive at 
\begin{align*}
F_{\rm crys}[\rho_n]&=\cF_{\rm crys}[\rho_n,Q_n]\\
&=\tr_0(H^0_{\rm per}-\epsilon_{\rm F})X_RQ_nX_R+\tr_0(H^0_{\rm per}-\epsilon_{\rm F})Y_RQ_nY_R
+D(\rho_Q,\rho) + D(\rho_{Y_R Q_n Y_R},\rho_n-\rho)\\
&\qquad+\frac12D(\rho_{X_RQX_R},\rho_{X_RQX_R})+D(\rho_{X_RQX_R},\rho_{Y_RQY_R})+\frac12D(\rho_{Y_RQ_nY_R},\rho_{Y_RQ_nY_R})
+\epsilon_n(R)+o(1)\\
&\geq \tr_0(H^0_{\rm per}-\epsilon_{\rm F})X_RQ_nX_R+F_{\rm crys}[\rho_n-\rho]
+D(\rho_Q,\rho) \\
&\qquad+\frac12D(\rho_{X_RQX_R},\rho_{X_RQX_R})+D(\rho_{X_RQX_R},\rho_{Y_RQY_R})+\epsilon_n(R) + o(1)
\end{align*}
where we have used that $Y_RQ_nY_R$ is an admissible trial state for $\crysf[\rho_n-\rho,Q]$ by Lemma~\ref{lem:abstract_sum}. Passing to the liminf and using Fatou's lemma yields
\begin{multline*}
\liminf_{n\to\ii}\left(F_{\rm crys}[\rho_n]-F_{\rm crys}[\rho_n-\rho]\right)\geq \tr_0(H^0_{\rm per}-\epsilon_{\rm F})X_RQX_R+D(\rho_Q,\rho)\\ 
+\frac12D(\rho_{X_RQX_R},\rho_{X_RQX_R})+D(\rho_{X_RQX_R},\rho_{Y_RQY_R})-\frac{C}{R}.
\end{multline*}
We have $\rho_{X_RQX_R}\to\rho_Q$ and  $\rho_{Y_RQY_R}\to0$ strongly in $\cC\cap L^2$, as $R\to\ii$ by Lemma \ref{lem:loc approx}. So, using Fatou's lemma again for the kinetic energy term and taking the limit $R\to\ii$, we arrive at the result
\begin{equation}
\liminf_{n\to\ii}\left(F_{\rm crys}[\rho_n]-F_{\rm crys}[\rho_n-\rho]\right)\geq \tr_0(H^0_{\rm per}-\epsilon_{\rm F})Q+D(\rho_Q,\rho)
+\frac12D(\rho_{Q},\rho_{Q})\geq F_{\rm crys}[\rho],
\label{eq:estim_below_Q}
\end{equation}
which is the lower bound corresponding to \eqref{eq:decouple}. 

\bigskip

\noindent\textsl{Step 2 : proof of Corollary \ref{cor:continu_Fcrys} with $\rho\equiv0$}. 

We pick a sequence $\rho_n \wto 0$ and denote by $(Q_n)$ the corresponding sequence of minimizers. Since $(\rho_n)$ is bounded in $\cC$, $(Q_n)$ is bounded in $\Q$ and, up to extraction, converges weakly to some $Q\in \Q$, which implies that $\rho_{Q_n} \wto \rho_Q$ weakly in $\cC$. We prove here that $Q\equiv0$. 

Thanks to the lower bound part of \eqref{eq:decouple} we have just proved, we write
\[
-\frac{1}{2} D(\rho_Q,\rho_Q) + \cryse[\rho_n] + o(1) \leq \cryse[- \rho_Q] + \cryse[\rho_n] + o(1) \leq \cryse[\rho_n - \rho_Q] \leq \cryse[\rho_n] - D(\rho_Q,\rho_{Q_n}), 
\]
using $Q_n$ as a trial state for $\cryse[\rho_n - \rho_Q]$ and the simple lower bound $\crysf[\nu,Q] \geq -\frac{1}{2} D(\nu,\nu)$. Taking the limit $n\to \infty$ we therefore obtain $D(\rho_Q,\rho_Q)=0$, which implies $Q=0$ by~\eqref{eq:estim_below_Q}.

\bigskip

\noindent\textsl{Step 3: Upper bound.}
We now construct a trial state for $\cryse[\rho_n]$ to obtain the upper bound part of \eqref{eq:decouple}. In previous works~\cite{HaiLewSer-09,CanDelLew-08a}, the special structure of the set of admissible states was used (see the Appendix of~\cite{HaiLewSer-09}). We propose here a new method based on Lemma \ref{lem:abstract_sum}.

Let $Q$ and $Q_n$ be two minimizers for, respectively, the problems $F_{\rm crys}[\rho]$ and $F_{\rm crys}[\rho_n-\rho]$. Recall that they must satisfy the constraint $-\FSm\leq Q,Q_n\leq1-\FSm$ and that, using Step 2, $Q_n\wto 0$. Let $\chi_R$ be a localization function of compact support as before, and $\eta_R=\sqrt{1-\chi_R^2}$. Consider the trial state (we use the notation of Lemma~3.8 with $\Pi=\FSm$)
$$Q_{n,R}:= X_RQX_R + Y_R Q_n Y_R.$$
We have $-\FSm\leq Q_{n,R}\leq1-\FSm$, by Lemma~\ref{lem:abstract_sum}, and thus 
$$ \cryse[\rho_n] \leq \crysf[\rho_n,Q_{n,R}].$$
We use that $X_RQ_nX_R$ and $Y_RQY_R$ both satisfy the constraint \eqref{eq:contrainte cristal}, hence that their kinetic energy is non-negative
$$\Tro \left(\left(\FSh-\FSl \right) X_RQ_nX_R \right)\geq0.$$
So we have for instance
\begin{align*}
\Tro \left(\left(\FSh-\FSl \right) Y_RQ_nY_R \right)&\leq \Tro \left(\left(\FSh-\FSl \right) X_RQ_nX_R \right)+\Tro \left(\left(\FSh-\FSl \right) Y_RQ_n Y_R\right)\\
&\leq\Tro \left(\left(\FSh-\FSl \right) Q_n \right)+\frac{C}{R^2}
\end{align*}
by~\eqref{eq:localization_kinetic}. We thus get
$$ \Tro \left(\left(\FSh-\FSl \right)  Q_{n,R} \right) \leq \Tro \left(\left(\FSh-\FSl \right) Q_n \right) + \Tro \left(\left(\FSh-\FSl \right)  Q \right)+\frac{C}{R^2}.$$ 
For the electrostatic terms we argue as in Step 1, using that $\rho_{X_RQ_nX_R}\to0$ strongly in $\cC$ and $\rho_{Y_RQ_nY_R}\wto0$ weakly in $\cC$ as $n\to\ii$, for fixed $R$: 
\begin{align*}
D(\rho_n,\rho_{Q_{n,R}}) &= D(\rho,\rho_{X_R Q X_R}) + D(\rho_n -\rho, \rho_{Y_R Q_n Y_R}) + D(\rho_n - \rho, \rho_{X_RQX_R}) + D(\rho,\rho_{Y_R Q_n Y_R})\\
&= D(\rho,\rho_{X_R Q X_R}) + D(\rho_n -\rho, \rho_{ Q_n }) - D(\rho_n -\rho, \rho_{X_RQ_nX_R}) + o(1) + \epsilon_n(R)\\
&=D(\rho,\rho_{X_R Q X_R}) + D(\rho_n -\rho, \rho_{ Q_n }) + o(1) + \epsilon_n(R)
\end{align*}
where we have used \eqref{eq:error_localization_rho} again. Similarly  
\begin{align*}
D(\rho_{Q_{n,R}},\rho_{Q_{n,R}}) &= D(\rho_{X_RQX_R},\rho_{X_RQX_R}) + D(\rho_{Y_RQ_nY_R},\rho_{Y_RQ_nY_R}) + 2 D(\rho_{X_RQX_R},\rho_{Y_RQ_nY_R})
\\ &= D(\rho_{X_RQX_R},\rho_{X_RQX_R}) + D(\rho_{Q_n},\rho_{Q_n}) + o(1)+ \epsilon_n (R)
\end{align*}
since $Q_n \wto 0$ and 
\begin{align*}
D(\rho_{Y_RQ_nY_R},\rho_{Y_RQ_nY_R})&= D(\rho_{Q_n},\rho_{Q_n}) - 2 D(\rho_{X_RQ_nX_R},\rho_{Q_n}) + D (\rho_{X_RQ_nX_R},\rho_{X_RQ_nX_R}) +\epsilon_n (R)\\
&= D(\rho_{Q_n},\rho_{Q_n}) + o(1) + \epsilon_n (R).
\end{align*}


Recalling that $\rho_{X_RQX_R} \to \rho_Q$ strongly in $\cC$ when $R\to \infty$, we can finally take first the limit $n\to \infty$ and then the limit $R\to\infty$ to conclude
\[
\limsup_{n\to \infty} \left(\cryse[\rho_n] - \cryse[\rho]- \cryse[\rho_n-\rho] \right) \leq 0 
\]
and the proof of Proposition~\ref{pro:decouple} is complete.

\bigskip

\noindent\textsl{Step 4: End of the proof of Corollary~\ref{cor:continu_Fcrys}.}
Let $(\rho_n)$ be any sequence such that $\rho_n\wto\rho$ weakly in $\cC$, and $Q_n$ be any associated sequence of minimizers for $\crysf[\rho_n,\cdot]$. Extracting a subsequence we may assume that $Q_n\wto Q$ in $\cQ$. Coming back to the lower bound~\eqref{eq:estim_below_Q} obtained in Step 1 and using~\eqref{eq:decouple}, we see that
$$\cryse[\rho]=\lim_{n\to\ii}\Big(\cryse[\rho_n]-\cryse[\rho_n-\rho]\Big)\geq \crysf[\rho,Q]\geq \cryse[\rho].$$
This shows that $Q$ minimizes $\crysf[\rho,\cdot]$ and concludes the proof of Corollary~\ref{cor:continu_Fcrys}. \hfill\qed

\section{Existence of polarons: Proof of Theorem \ref{theo:E_1}}\label{sec:proof1}
Before turning to the more complicated case of $N$ particles for which we have to adapt Theorem 25 in \cite{Lewin-11}, we deal with the simpler one-particle case. The proof that a minimizer always exists for one particle follows from usual techniques of nonlinear analysis. In this context the most difficult is to verify the one-particle binding inequality~\eqref{eq:E_1}, which we do first.

\subsubsection*{Step 1. Proof of the one-particle binding inequality}
The aim of this first step is to prove the following important

\begin{lemma}[One-particle binding]\label{lem:E_1_no_vanishing}\mbox{}\\
We have
\begin{equation}
E(1)<E_{\rm per}:=\inf\sigma\left(-\frac{\Delta}{2m}+V^0_{\rm per}\right).
\label{eq:upper_bound_E_1}
\end{equation} 
\end{lemma}
\begin{proof}
Let $u_{\rm per}$ denote the first $\Ll$-periodic eigenfunction of $-{\Delta}/(2m)+V^0_{\rm per}$, which is a solution of 
\begin{equation}
\left(-\frac{\Delta}{2m}+V^0_{\rm per}\right)u_{\rm per}=E_{\rm per}\,u_{\rm per}. 
\label{eq:u_per}
\end{equation}
We assume that $u_{\rm per}$ is normalized, $\int_{\Gamma}|u_{\rm per}|^2=1$ where $\Gamma$ is the unit cell of $\Ll$. Since $u_{\rm per}\in H^2_{\rm per}(\Gamma)$, we have also $\nu_{\rm per}:=|u_{\rm per}|^2\in H^2_{\rm per}(\Gamma)$. The Fourier coefficients $(\widehat{\nu}_{\rm per}(k))_{k\in\Ll^*}$ thus satisfy $(|k| ^2 \widehat{\nu}_{\rm per}(k))_{k\in\Ll^*}\in \ell ^2 (\Ll^*)$ and consequently belong to $\ell^1(\Ll^*)$ :
\begin{equation}
\sum_{k\in\Ll^*}\left|\widehat{\nu}_{\rm per}(k)\right|<\ii.
\label{eq:ell_1} 
\end{equation}
Here $\Ll^*$ is the dual lattice of $\Ll$, whose unit cell will be denoted by $\Gamma^*$.
We can write
$$|u_{\rm per}(x)|^2=\frac{1}{|\Gamma^*|}\sum_{k\in \Ll^*}\widehat{\nu}_{\rm per}(k)\,e^{ik\cdot x}$$

Consider now a fixed function $\chi\in C^\ii_c(\R^3)$ such that $\int|\chi|^2=1$, and define the following test function for the variational problem $E(1)$:
\begin{equation}\label{eq:trial pouet}
\psi_\lambda:=u_{\rm per}(x)\,\chi_\lambda(x),\qquad \text{with}\quad\chi_\lambda(x):=\lambda^{-3/2}\chi\left(\frac{x}{\lambda}\right).
\end{equation}
The corresponding density is
\begin{equation}
|\psi_\lambda(x)|^2:=|u_{\rm per}(x)|^2\;\left|\chi_\lambda(x)\right|^2=\frac{1}{|\Gamma^*|}\sum_{k\in \Ll^*}\widehat{\nu}_{\rm per}(k)\,e^{ik\cdot x}\;\left|\chi_\lambda(x)\right|^2.
\end{equation}
Remark that
\begin{align*}
D\left(|\chi_\lambda|^2e^{ik\cdot},|\chi_\lambda|^2e^{ik\cdot}\right)&=4\pi\lambda^{-3}\int_{\R^3}\frac{\left|\widehat{|\chi|^2}(p)\right|^2}{|p/\lambda+k|^2}\,dp\underset{\lambda\to\ii}{\sim}\frac{4\pi}{\lambda^{3}|k|^{2}}\int_{\R^3}\left|\widehat{|\chi|^2}(p)\right|^2\,dp
\end{align*}
for any $k\in \Ll ^*\setminus \left\lbrace 0 \right\rbrace$. Using~\eqref{eq:ell_1} and the fact that $u_{\rm per}$ is normalized, we deduce that
$$\norm{|\psi_\lambda|^2-|\chi_\lambda|^2}_{\cC}=O\left(\frac{1}{\lambda^{3/2}}\right).$$
Similarly, the normalization factor is
$$\int_{\R^3}|u_{\rm per}(x)|^2|\chi_\lambda(x)|^2\,dx=\frac{1}{(2\pi)^{3/2}}\sum_{k\in \Ll^*}\widehat{\nu}_{\rm per}(k)\,\widehat{\left|\chi\right|^2}(\lambda k)=1+O\left(\frac{1}{\lambda^p}\right)$$
for all $p\in\N$.
Of course, we have by scaling
$$D\left( |\chi_\lambda|^2,|\chi_\lambda|^2\right)=\frac{1}{\lambda}D\left( |\chi|^2,|\chi|^2\right).$$
We deduce from all this that
$$F_{\rm crys}\left[\frac{|\psi_\lambda|^2}{\int_{\R^3}|\psi_\lambda|^2}\right]=F_{\rm crys}\left[|\chi_\lambda|^2\right]+O\left(\frac{1}{\lambda^{3/2}}\right)$$
by~\eqref{eq:loc_Lipschitz}. In Theorem 1.4 of~\cite{LewRou-11} we have studied in detail the behavior of the crystal energy when the external density is very spread out. We have proved that
\begin{equation}
F_{\rm crys}\left[|\chi_\lambda|^2\right]=F_{\rm crys}\left[\lambda^{-3}|\chi(\cdot/\lambda)|^2\right]=\frac{1}{\lambda}\,F^{\rm P}_{\epsilon_{\rm M}}\left[|\chi|^2\right]+o\left(\frac{1}{\lambda}\right)
\end{equation}
where $F^{\rm P}_{\epsilon_{\rm M}}$ is Pekar's effective interaction energy
$$F^{\rm P}_{\epsilon_{\rm M}}\left[\rho\right]:=2\pi\int_{\R^3}|\widehat{\rho}(p)|^2\left(\frac{1}{p^T\epsilon_{\rm M}p}-\frac{1}{|p|^2}\right)\,dp.$$
Since $\varepsilon_{\rm M}>1$, we have $F^{\rm P}_{\epsilon_{\rm M}}\left[\rho\right]<0$ for all $\rho$. So the exact (first order) behavior of the crystal energy for our trial state is 
$$F_{\rm crys}\left[\frac{|\psi_\lambda|^2}{\int_{\R^3}|\psi_\lambda|^2}\right]=\frac{F^{\rm P}_{\epsilon_{\rm M}}\left[|\chi|^2\right]}{\lambda}+o_{\lambda\to\ii}\left(\frac{1}{\lambda}\right).$$
The two other terms in the energy $\cE$ are easier to handle. A simple computation based on the equation~\eqref{eq:u_per} of $u_{\rm per}$ shows that
$$\int_{\R^3}\frac1{2m}|\nabla\psi_\lambda|^2+V^0_{\rm per}|\psi_\lambda|^2=E_{\rm per}\int_{\R^3}|\psi_\lambda|^2+\frac1{2m}\int_{\R^3}|u_{\rm per}|^2|\nabla\chi_\lambda|^2$$
(see Lemma 2.2 in~\cite{LewRou-11}). Of course,
$$\int_{\R^3}|u_{\rm per}|^2|\nabla\chi_\lambda|^2\leq C\int_{\R^3}|\nabla\chi_\lambda|^2=\frac{C}{\lambda^2}\int_{\R^3}|\nabla\chi|^2$$
since $u_{\rm per}\in H^2_{\rm per}\subset L^\ii(\R^3)$. 
As a conclusion we have shown that
$$\cE\left(\frac{\psi_\lambda}{\sqrt{\int_{\R^3}|\psi_\lambda|^2}}\right)=E_{\rm per}+\frac{F^{\rm P}_{\epsilon_{\rm M}}\left[|\chi|^2\right]}{\lambda}+o_{\lambda\to\ii}\left(\frac{1}{\lambda}\right).$$
Since $F^{\rm P}_{\epsilon_{\rm M}}\left[|\chi|^2\right]<0$, the inequality~\eqref{eq:upper_bound_E_1} follows.
\end{proof}

\begin{remark}
Note that the proof of the above lemma actually uses a construction reminiscent of a large polaron: the trial state \eqref{eq:trial pouet} describes a particle extended over a region much larger than the lattice spacing, in the spirit of \cite{LewRou-11}. This is of course only a trial state argument, and the ground state, when it exists, lives itself on a smaller scale.
\end{remark}


\subsubsection*{Step 2. Compactness of minimizing sequences and existence of a minimizer for $N=1$}
We now turn to the proof of the other statements in Theorem~\ref{theo:HVZ} dealing with the one-particle case $E(1)$.

Let $(\psi_n)$ be a minimizing sequence for $E(1)$. Since $0\geq F_{\rm crys}[|\psi|^2]\geq -D(|\psi|^2,|\psi|^2)/2$, it is easy to see that $(\psi_n)$ is bounded in $H^1(\R^3)$. We define the largest mass that subsequences can have up to translations by
$$M:=\sup\left\{\int_{\R^3}|\psi|^2\ :\ \exists(x_k)\subset\R^3,\ \psi_{n_k}(\cdot-x_k)\wto\psi\text{ weakly in $H^1(\R^3)$}\right\}.$$
We know~\cite{Lions-84} that $M=0$ if and only if $\psi_n\to0$ strongly in $L^p(\R^3)$ for all $2<p<6$, a phenomenon that is usually called \emph{vanishing}. But if this is the case, we get $\norm{|\psi_n|^2}_{\cC}\to0$ by the Hardy-Littlewood-Sobolev inequality, and therefore $F_{\rm crys}[|\psi_n|^2]\to0$ by~\eqref{eq:loc_Lipschitz}. We then get $E(1)\geq E_{\rm per}:=\inf\sigma(H^0_{\rm per})$ which is impossible by Lemma~\ref{lem:E_1_no_vanishing}. Thus $M>0$.

Since $M>0$ we can find a subsequence (denoted the same for simplicity), such that $\psi_n(\cdot-x_n)\wto\psi\neq0$. 
We can of course write $x_n=k_n+y_n$ where $k_n\in\Ll$ and $y_n\in\Gamma$. Extracting subsequences again we get $y_n\to y\in\Gamma$, the unit cell of the lattice $\Ll$. Therefore $\psi_n(\cdot-k_n)\wto \psi(\cdot+y)\neq0$. Since our energy functional is invariant under the translations of $\Ll$, the new sequence $\psi_n(\cdot-k_n)$ is again a minimizing sequence. Without loss of generality we can thus assume that $\psi_n\wto\psi\neq0$. Now, if we can prove that $\int_{\R^3}|\psi|^2=1$, we will get strong convergence in $L^2$ and it is then standard to upgrade this to strong convergence in $H^1$. We argue by contradiction and assume that $0<\int_{\R^3}|\psi|^2<1$.

We will now show that the energy decouples into two pieces. Since $\psi_n \wto \psi$ in $H^1 (\R ^3)$ we may assume that $|\psi_n| ^2 \wto |\psi| ^2$ in $\coul$.  We then use that, by~\eqref{eq:decouple} in Proposition~\ref{pro:decouple},
$$F_{\rm crys}[|\psi_n|^2]\geq F_{\rm crys}[|\psi|^2]+F_{\rm crys}[|\psi_n|^2-|\psi|^2]+o(1).$$
Note that
$$|\psi_n|^2-|\psi|^2-|\psi_n-\psi|^2=2\Re\overline{\psi}(\psi_n-\psi)\to0$$
strongly in $L^1(\R^3)$ (we use here that $\psi_n\to\psi$ strongly in $L^2_{\rm loc}$ and an $\epsilon/2$ argument), hence in $L^{6/5}(\R^3)$ by interpolation. Thus
$$\lim_{n\to\ii}\bigg|F_{\rm crys}[|\psi_n|^2-|\psi|^2]-F_{\rm crys}[|\psi_n-\psi|^2]\bigg|=0$$
by~\eqref{eq:loc_Lipschitz} in Lemma~\ref{lem:subcritical}, and we arrive at
$$F_{\rm crys}[|\psi_n|^2]\geq F_{\rm crys}[|\psi|^2]+F_{\rm crys}[|\psi_n-\psi|^2]+o(1).$$
On the other hand, it is clear from the weak convergence $\psi_n\wto\psi$ in $H^1(\R^3)$ (and from the fact that the form domain of $-{\Delta}/{(2m)} + \FSp$ is $H^1(\R^3)$), that
$$\pscal{\psi_n,\left(-\frac{\Delta}{2m} + \FSp\right)\psi_n}=\pscal{\psi,\left(-\frac{\Delta}{2m} + \FSp\right)\psi}+\pscal{(\psi_n-\psi),\left(-\frac{\Delta}{2m} + \FSp\right)(\psi_n-\psi)}+o(1).$$
Hence we have shown that
$$\cE[\psi_n]\geq \cE[\psi]+\cE[\psi_n-\psi]+o(1).$$
Now we use that $F_{\rm crys}$ is concave to infer
$$F_{\rm crys}\left[|\psi_n-\psi|^2\right]\geq\left(\int_{\R^3}|\psi_n-\psi|^2\right)F_{\rm crys}\left[\frac{|\psi_n-\psi|^2}{\int_{\R^3}|\psi_n-\psi|^2}\right],$$
leading to 
\begin{align*}
\cE[\psi_n]&\geq \cE[\psi]+\left(\int_{\R^3}|\psi_n-\psi|^2\right)\cE\left[\frac{\psi_n-\psi}{\sqrt{\int_{\R^3}|\psi_n-\psi|^2}}\right]+o(1)\\
&\geq \cE[\psi]+\left(\int_{\R^3}|\psi_n-\psi|^2\right)E(1)+o(1). 
\end{align*}
Passing to the limit $n\to\ii$, we find
$$E(1)\geq \cE[\psi]+\left(1-\int_{\R^3}|\psi|^2\right)E(1).$$
It is now time to use the strict concavity at the origin~\eqref{eq:strict concavity}
$$F_{\rm crys}\left[|\psi|^2\right]>\left(\int_{\R^3}|\psi|^2\right)F_{\rm crys}\left[\frac{|\psi|^2}{\int_{\R^3}|\psi|^2}\right],$$
which yields
$$\cE[\psi]>\left(\int_{\R^3}|\psi|^2\right)\cE\left[\frac{\psi}{\sqrt{\int_{\R^3}|\psi|^2}}\right]\geq\left(\int_{\R^3}|\psi|^2\right)E(1).$$
Therefore we have proved that $E(1)>E(1)$ which is a contradiction, unless $\int_{\R^3}|\psi|^2=1$. This concludes the proof in the case of one particle.

\section{Binding of $N$-polarons: Proof of Theorem \ref{theo:HVZ}}\label{sec:proof}

We now turn to the case of $N\geq2$. With the input of Section \ref{sec:prop crystal}, the proof more or less follows that of Theorem 25 in \cite{Lewin-11}. We nevertheless sketch the main steps for the convenience of the reader.

We will denote 
\[
 H(N):=\sum_{j=1} ^N \left( - \frac{\Delta_j}{2m} + \FSp (x_j)\right) + \sum_{i<j}\frac{1}{|x_i-x_j|}. 
\]
In order to relate problems with different particle numbers to one another, it is crucial to introduce the antisymmetric truncated Fock space
\[
 \FockN = \bigoplus_{n=0} ^N \bigwedge_{i=1} ^n L^2(\R ^3)  
\]
where $\bigwedge$ is the antisymmetric tensor product and we use the convention $\bigwedge_{i=1} ^0 L^2(\R ^3)=\C$. A state on $\FockN$ is an operator $\Gamma \in \Sch ^1 (\FockN)$ with $\Tr (\Gamma) = 1$. In the sequel we restrict ourselves to states commuting with the number operator
\[
\NN = \bigoplus_{n=0} ^N n. 
\]
This means (see \cite{Lewin-11}, Remark 7) that they take the form
\begin{equation}\label{eq:state generic}
\Gamma = G_{00}\oplus\ldots \oplus G_{NN} 
\end{equation}
with $G_{ii} \in \Sch ^1 \left( \bigwedge_{i=1} ^n L^2(\R ^3)  \right) $. We denote by
\[
\mathbb{H} = \bigoplus _{n= 0} ^N H(n)
\]
the many-body second-quantized Hamiltonian.
To any state $\Gamma$ are associated a density $\rho_{\Gamma}\in L^1 (\R ^3)$, one-body density matrix $[\Gamma] ^{1,1}\in \Sch ^1 (L ^2 (\R ^3))$ and two-body density matrix $[\Gamma] ^{2,2}\in \Sch ^1 (L ^2 (\R ^3) \times L ^2 (\R ^3))$ (see \cite{Lewin-11}, Section 1). We can extend the energy to Fock space as
\begin{eqnarray*}
\E[\Gamma]&=& \Tr_{\FockN}\left(\mathbb{H} \Gamma \right) + \cryse[\rho_{\Gamma}]\\
&=& \Tr_{L ^2 (\R ^3)}\left( \left(- \Delta + \FSp \right) [\Gamma] ^{1,1} \right) + \Tr_{L ^2 (\R ^3) \times L ^2 (\R ^3)}\left( W [\Gamma] ^{2,2} \right) + \cryse[\rho_{\Gamma}]
\end{eqnarray*}
where $W$ acts on $L ^2 (\R ^3) \times L ^2 (\R ^3)$ as the multiplication by $|x-y|^{-1}$. For a pure state $\Gamma = 0\oplus \ldots \oplus \ketl \Psi \ketr\bral \Psi\brar$ with $\Psi \in L ^2 (\R ^{3N})$ one can check that $\E[\Gamma] = \E [\Psi]$. More generally, for a state of the form \eqref{eq:state generic}, we have
\[
\E[\Gamma] = \sum_{n=1} ^{N} \Tr_{\bigwedge_1^nL^2(\R^3)} \left( H(n) G_{nn} \right) + \cryse\left[\sum_{n=1} ^N \rho_{G_{nn}}\right]. 
\]

\subsubsection*{Step 1. Large binding inequality.}
We claim that 
\begin{equation}\label{eq:binding large}
 E(N) \leq E(N-k) + E(k) \mbox{ for all } k=1,\ldots, N-1. 
\end{equation}
To see this, we consider the following trial state:
\begin{equation}\label{eq:trial state}
\Psi ^N_R:= \Psi ^{N-k} \wedge \Psi ^{k} \left( .-R \vec{\tau} \right)
\end{equation}
where $(\Psi ^{N-k})$ and $(\Psi ^{k})$ are compactly supported fixed trial states for $E(N-k)$ and $E(k)$ respectively, $\vec{\tau}\in\Ll$ is a lattice translation and $R\in\N$ is large enough for $\rho_{\Psi ^{N-k}}$ and $\rho_{\Psi ^{k}} \left( .-R \vec{\tau} \right)$ to have disjoint supports. The symbol $\wedge$ denotes the antisymmetric tensor product. We first take the limit $R\to \infty$ to obtain
\begin{equation}\label{eq:binding large -} 
E(N) \leq \E[\Psi ^{N-k}] + \E[\Psi ^{k}].
\end{equation}
Optimizing then with respect to $\Psi^{N-k}$ and $\Psi^k$ concludes the proof of \eqref{eq:binding large}. To see that \eqref{eq:binding large -} holds, we note that by construction
\[
\rho_{\Psi^N_R} =  \rho_{\Psi ^{N-k}} + \rho_{\Psi  ^{k}} \left( .-R \vec{\tau} \right)
\]
for large enough $R$, thus we can use Proposition \ref{pro:decouple} and take the limit $R\to \infty$ with $R\in\N$ to obtain 
\[
\lim_{R\to\ii}\cryse[\rho_{\Psi ^N_R}] = \cryse[\rho_{\Psi ^{N-k}}] +\cryse[\rho_{\Psi  ^{k}}].
\]
The other terms in the energy can be treated as usual to obtain \eqref{eq:binding large -}. 

Note that the argument here also proves by contradiction that Item $(2)$ of Theorem \ref{theo:HVZ} implies Item~$(1)$. If there is equality in \eqref{eq:binding large}, we can choose $\Psi_n^{N-k}$ and $\Psi_n^k$ minimizing sequences for $E(N-k)$ and $E(k)$ respectively and, taking $R_n\to\ii$ very fast, we obtain a minimizing sequence for $E(N)$ that is not precompact, even up to translations because some mass is lost at infinity.

%

\subsubsection*{Step 2. Absence of vanishing.}
We consider a minimizing sequence $(\Psi_n )$ for $E(N)$ and denote by $\Gamma_n = 0\oplus \ldots \ldots \oplus \ketl \Psi_n \ketr \bral \Psi_n \brar $ the associated state in the truncated antisymmetric Fock space. It is easy to see, using in particular Lemma \ref{lem:subcritical} that $(\Psi_n)$ is bounded in $H ^1 (\R ^{3N})$. As in the one-body case treated before, we define a criterion for the vanishing of the minimizing sequence. We use the concept of geometric convergence (see Section 2 in \cite{Lewin-11} for the definition). We look at the the mass of the possible geometric limits, up to translations and extraction, of $(\Gamma_n)$   
\[
M:= \sup \left\lbrace \Tr \left( \NN \Gamma \right),\: \exists \vec{v}_k\subset\R^3,\ \vec{v}_k \Gamma_{n_k} \vec{v}_k ^* \rightharpoonup_g \Gamma \right\rbrace
\]
where we recall that $\NN$ is the number operator in Fock space. As explained in \cite{Lewin-11}, Lemma 24, if $M=0$  then $\rho_{\Psi_n }\to 0$ strongly in $L ^p (\R ^3) $ for all $1 < p < 3$. Using then Lemma \ref{lem:subcritical} we obtain $F_{\rm crys}[\rho_{\Psi_n}]\to0$ and therefore 
\[
 E(N) = \lim_{n\to \infty}\E [\Psi_n] \geq \inf \sigma \left(\tilde{H}(N) \right)= N E_{\rm per}
\]
where 
\begin{equation}\label{eq:hamiltonien}
\tilde{H}(N):= \sum_{j=1} ^N \left( - \frac{\Delta_j}{2m} + \FSp (x_j)\right)=\sum_{j=1} ^N \left(H^0_{\rm per}\right)_{x_i}. 
\end{equation}
Note that, by induction on $N$, \eqref{eq:binding large} implies
$ E(N) \leq N E(1)$.
We have already seen in~\eqref{eq:upper_bound_E_1} above that $E(1) < \inf \sigma \left( H^0_{\rm per}\right) $. Hence we reach a contradiction and conclude that $M>0$. 
 
\subsubsection*{Step 3. Decoupling via localization}
Since $M>0$ (and arguing as in the previous section) we have, up to the extraction of a subsequence, 
$\vec{v}_n \Gamma_{n} \vec{v}_n ^* \togeo \Gamma$
with $\Tr(\NN \Gamma) >0$ and where $(\vec{v}_n)\subset\Ll$ is a sequence of lattice translations. Using the invariance of the energy, Lemma \ref{lem:translation}, we can thus assume that our minimizing sequence satisfies
\begin{equation}\label{eq:conv geo}
\Gamma_n  \togeo \Gamma
\end{equation}
with $\Tr(\NN \Gamma) >0$. Also we have
$\sqrt{\rho_{\Gamma_n}}\wto\sqrt{\rho_\Gamma}$
weakly in $H^1(\R^3)$ and strongly in $L^2_{\rm loc}$. Also $\rho_{\Gamma_n}\wto\rho_\Gamma$ in the Coulomb space $\cC$ and we immediately deduce by~\eqref{eq:decouple} that
$$F_{\rm crys}[\rho_{\Gamma_n}]\geq F_{\rm crys}[\rho_{\Gamma}]+F_{\rm crys}[\rho_{\Gamma_n}-\rho_\Gamma]+o(1).$$

We now pick a sequence of radii $R_n \to \infty$ and define smooth localization functions $\chin$ and $\etan$ such that $\chin ^2 + \etan ^2 = 1 $, $\supp(\chin) \subset B(0,2R_n)$ and $\supp(\chin) \subset \R ^3 \setminus B(0,3R_n)$. For any bounded operator $B$ (in particular the multiplication by a function $\chi$) on $L ^2(\R^3)$ such that $0 \leq BB^* \leq 1$ we will denote by $(\Gamma)_{B}$ the $B$-localization of a state $\Gamma$, as defined in \cite{Lewin-11}, Section 3. Of importance to us will be the following properties of localization:
\begin{eqnarray}\label{eq:prop localisation}
\rho_{\Gamma_{\chi}} &=& \chi^2 \rho_{\Gamma} \nonumber \\
 \: [\Gamma _\chi ] ^{1,1} &=& \chi [\Gamma] ^{1,1} \chi \nonumber \\
 \: [\Gamma _\chi ] ^{2,2} &=& \chi\otimes \chi\, [\Gamma] ^{2,2} \chi \otimes \chi.
\end{eqnarray}
Also, for a state of the form \eqref{eq:state generic}, writing
\[
(\Gamma)_{\chin} =  G_{0} ^{\chin}\oplus\ldots \oplus G_{N} ^{\chin}, \: (\Gamma)_{\etan} =  G_{0} ^{\etan}\oplus\ldots \oplus G_{N} ^{\etan},
\]
the condition $\chin ^2 + \etan ^2 = 1 $ implies the relation
\begin{equation}\label{eq:fundamental}
 \Tr \left( G_j ^{\chin}\right) = \Tr \left( G_{N-j} ^{\etan}\right).
\end{equation}

Using concentration functions as in Step 4 of the proof of \cite{Lewin-11}, Theorem 25 we have, extracting a further subsequence if necessary
\begin{equation}\label{eq:conv forte}
\left( \Gamma_n \right)_{\chin} \to \Gamma \mbox{ strongly in } \Sch ^1 \left( \FockN \right)
\end{equation}
and
\begin{equation}\label{eq:conv forte densite}
(\chin)^2 \rho_{\Gamma_n}\to \rho_\Gamma \mbox{ strongly in $L^p(\R^3)$ for all $2\leq p<3$.}
\end{equation}
Using~\eqref{eq:loc_Lipschitz}, this can be used to prove that 
$$F_{\rm crys}\left[\rho_{\Gamma_n}- \rho_\Gamma\right]=F_{\rm crys}\left[(\etan)^2\rho_{\Gamma_n}\right]+o(1).$$
Thus
$$F_{\rm crys}\left[\rho_{\Gamma_n}\right]\geq F_{\rm crys}\left[\rho_\Gamma\right]+F_{\rm crys}\left[\rho_{(\Gamma_n)_{\etan}}\right]+o(1).$$

We have seen that the nonlinear energy $F_{\rm crys}$ decouples. The other terms are treated following~\cite{Lewin-11}. For the one-particle part we use the IMS formula
\[
 \Delta = \chin \Delta \chin + \etan \Delta \etan + |\nabla \chin| ^2  + |\nabla \etan| ^2
\]
to obtain (we use \eqref{eq:prop localisation})
\begin{align*}
&\Tr \left( \left( -\frac{\Delta}{2m} + \FSp \right) [\Gamma_n] ^{(1,1)}\right)\\
&\qquad\quad\geq \Tr \left( \chin \left(  -\frac{\Delta}{2m} + \FSp \right)\chin [\Gamma_n] ^{(1,1)}\right)+\Tr \left( \etan \left(  -\frac{\Delta}{2m} + \FSp \right)\etan [\Gamma_n] ^{(1,1)}\right) - \frac{CN}{R_n ^2}\\
&\qquad\quad= \Tr \left( \left(  -\frac{\Delta}{2m} + \FSp \right)[(\Gamma_n)_{\chin}] ^{(1,1)}\right)+\Tr \left( \left(  -\frac{\Delta}{2m} + \FSp \right) [(\Gamma_n)_{\etan}] ^{(1,1)}\right) - \frac{CN}{R_n ^2}.
\end{align*}
The Coulomb interaction is treated exactly as in \cite{Lewin-11} and we conclude
\[
\left\langle \Psi_n, H(N) \Psi_n \right\rangle \geq  \Tr \left( \mathbb{H} (\Gamma_n)_{\chin}\right) +\Tr \left( \mathbb{H} (\Gamma_n)_{\etan}\right) + o(1).
\]
Using Fatou's lemma as well as the strong convergence of $(\chin)^2\rho_{\Gamma_n}$, we finally get
\begin{equation}
\cE[\Psi_n]\geq \cE\left[\Gamma\right]+\cE\left[(\Gamma_n)_{\etan}\right]+o(1).
\end{equation}
which is the desired decoupling of the energy.

\subsubsection*{Step 4. Conclusion.}
The rest of the argument follows exactly \cite{Lewin-11}. Writing the geometric limit of $\Gamma_n$
\[
\Gamma = G_{00}\oplus \ldots \oplus G_{NN},
\]
and using the concavity of $\cryse$, the fundamental relation \eqref{eq:fundamental} as well as the convergence \eqref{eq:conv forte}, we arrive at
\[
E(N)\geq \sum_{j=0} ^N \Tr(G_{jj})\left( E(j)+ E(N-j)\right). 
\]
Assuming the strict binding inequalities~\eqref{eq:binding}, this is possible only when $G_{11}=\ldots=G_{N-1N-1} = 0$. Hence we necessarily have $G_{NN}\neq 0$, otherwise we would obtain a contradiction with the fact that $\tr(\mathcal{N}\Gamma)>0$. 

To conclude, it is then enough to prove that $G_{00} = 0$, which is an easy consequence of the strict concavity \eqref{eq:strict concavity} of $\cryse$ (see Step 5 of the proof of Theorem 25 in \cite{Lewin-11} for details). We deduce that $\Tr(G_{NN}) = 1 = \Tr(\ketl\Psi_n\ketr \bral \Psi_n \brar)$, hence that the weak-$\ast$ convergence of $\ketl\Psi_n\ketr \bral \Psi_n \brar$ in $\Sch^1 (L^2 (\R ^3)) $to $G_{NN}$ is actually strong because no mass is lost in the weak limit. As $G_{NN}=\ketl\Psi\ketr \bral \Psi \brar$ where $\Psi$ is the weak limit of $\Psi_n$, we conclude that $\Psi_n$ converges to $\Psi$ strongly in $L^2 (\R ^3)$. The convergence in $H^1 (\R ^3)$ follows by standard arguments. 
\hfill \qed

\section*{Appendix. Proof of Lemma~\ref{lem:loc_properties}}\addcontentsline{toc}{section}{Appendix A}
We follow ideas of~\cite{CanDelLew-08a}. In the sequel we assume that $Q$ is finite rank, very smooth and decays fast enough, in order to justify the calculations. The conclusions for general $Q$ then follow by density, using Lemma 2 and Corollary 3 in \cite{CanDelLew-08a}.

\subsubsection*{Proof of \eqref{eq:uniform_cQ}: uniform bounds in $\cQ$}
The argument is the same for the terms involving $X_R$ and those involving $Y_R$, we thus discuss only the former. Recalling the definition \eqref{eq:crys space} of the space $\cQ$ and the fact that $|\FSh - \FSl| ^{-1/2} \left( 1-\Delta\right) ^{1/2}$ is uniformly bounded in operator norm (Lemma 1 in \cite{CanDelLew-08a}), our task is to estimate the terms $\big(X_RQX_R)^{\pm\mp}|\FSh-\FSl|^{1/2}$ in the Hilbert-Schmidt norm and $|\FSh-\FSl|^{1/2}\big(X_RQX_R)^{\pm\pm}|\FSh-\FSl|^{1/2}$ in the trace norm. We write
\begin{align*}
\big(X_RQX_R)^{+-}|\FSh-\FSl|^{1/2}&=\FSmp\chi_RQ^{+-}\chi_R\FSm|\FSh-\FSl|^{1/2}\\
&=\FSmp\chi_RQ^{+-}|\FSh-\FSl|^{1/2}\chi_R\FSm\\
&+\FSmp\chi_RQ^{+-}|\FSh-\FSl|^{1/2}|\FSh-\FSl|^{-1/2}\big[\chi_R,|\FSh-\FSl|^{1/2}\big]\FSm
\end{align*}
and deduce 
$$
\norm{(X_RQX_R)^{+-}|\FSh-\FSl|^{1/2}}_{\Sch ^2}\leq C \norm{Q^{+-}|\FSh-\FSl|^{1/2}}_{\Sch ^2}\leq C \norm{Q}_{\cQ}\\
$$
using that $\norm{\chi_R}_{L^\ii}=1$ and that
$$\norm{|\FSh-\FSl|^{-1/2}\big[\chi_R,|\FSh-\FSl|^{1/2}\big]}=O(R^{-1})$$
as shown in the proof of Lemma 11 in~\cite{CanDelLew-08a}. With similar computations, using that $|\FSh-\FSl|^{1/2}Q^{++}|\FSh-\FSl|^{1/2}\in \Sch ^1$ we obtain
$$
\norm{|\FSh-\FSl|^{1/2}\big(X_RQX_R)^{++}|\FSh-\FSl|^{1/2}}_{\Sch ^1} \leq C \norm{Q}_{\cQ}.
$$
The terms involving $(X_RQX_R)^{++}$ and $(X_RQX_R)^{-+}$ are estimated in exactly the same way. Finally, it was shown in Proposition 1 of~\cite{CanDelLew-08a} that the map $Q\in\cQ\mapsto \rho_Q\in L^2\cap \cC$ is continuous, hence the estimates on $\rho_{X_RQX_R}$ and $\rho_{Y_RQY_R}$ also follow.

\subsubsection*{Proof of \eqref{eq:error_localization_rho}: localization of the density}
We argue by duality, noting that
$$\int_{\R^3}\big(\rho_Q-\rho_{X_RQX_R}-\rho_{Y_RQY_R}\big)V=\tr\big(Q(V-X_RVX_R-Y_RVY_R)\big).$$
Inspired by the IMS formula, we now use that
\begin{align}
V&=\frac12(1-X_R^2-Y_R^2)V+\frac12V(1-X_R^2-Y_R^2) +\frac{X_R^2+Y_R^2}2V+V\frac{X_R^2+Y_R^2}2\nonumber\\
&=\frac12(1-X_R^2-Y_R^2)V+\frac12V(1-X_R^2-Y_R^2)+X_RVX_R+Y_RVY_R+\frac12[X_R,[X_R,V]]+\frac12[Y_R,[Y_R,V]].\label{eq:IMS-V}
\end{align}
The idea here is that $X^2_R+Y_R^2\simeq \chi_R^2+\eta_R^2=1$ which, unfortunately, is only true in the operator norm.

We start with the estimate on $\tr(Q(1-X_R^2-Y_R^2)V)$ (the second term is treated in the same way). We write as usual $Q=Q^{++}+Q^{--}+Q^{-+}+Q^{+-}$ and estimate each term separately. Recall that $X_R$ and $Y_R$ commute with $\FSm$, so we get for instance
$$
\tr(Q^{+-}(1-X_R^2-Y_R^2)V)=\tr(Q^{+-}(1-X_R^2-Y_R^2)\FSm V\FSmp)
= \tr(Q^{+-}(1-X_R^2-Y_R^2)[\FSm,V]).
$$
A bound from Lemma 5 in \cite{CanDelLew-08a} tells us that, if $V=V_1+V_2$ with $V_1\in \dot{H}^1(\R^3)$ and $V_2\in L^2(\R^3)$,
$$\norm{[\FSm,V]}_{\gS^2}\leq C(\norm{\nabla V_1}_{L^2}+\norm{V_2}_{L^2}).$$
We thus get
$$|\tr(Q^{+-}(1-X_R^2-Y_R^2)V)|\leq C\norm{1-X_R^2-Y_R^2}\norm{Q^{+-}}_{\gS^2}(\norm{\nabla V_1}_{L^2}+\norm{V_2}_{L^2}).$$
Finally recall that
$$\norm{1-X_R^2-Y_R^2}=\norm{\chi_R^2-X_R^2+\eta_R^2-Y_R^2}\leq \norm{\chi_R^2-X_R^2}+\norm{\eta_R^2-Y_R^2} \leq CR^{-1}$$
because
$$\chi_R^2-X_R^2=\chi_R^2-\FSm\chi_R\FSm\chi_R\FSm-(\FSm)^\perp\chi_R(\FSm)^\perp\chi_R(\FSm)^\perp$$
and the commutator $[\FSm,\chi_R]$ is known to be of order $O(R^{-1})$ in operator norm by Lemma 10 in~\cite{CanDelLew-08a}. The term involving $Y_R$ and $\etaR$ is treated in the same way. 
Therefore we have proved that
$$|\tr(Q^{+-}(1-X_R^2-Y_R^2)V)|\leq CR^{-1}\norm{Q^{+-}}_{\gS^2}(\norm{\nabla V_1}_{L^2}+\norm{V_2}_{L^2}).$$
For $Q^{--}$, we do not have a commutator but we can use the trace-class norm. We write
$$\tr(Q^{--}(1-X_R^2-Y_R^2)V)=\tr(Q^{--}(1-X_R^2-Y_R^2)\FSm V\FSm)$$
and estimate
\begin{eqnarray*}
\norm{\FSm V\FSm}&\leq& \norm{\FSm|H^0_{\rm per}-\epsilon_{\rm F}|}\norm{|H^0_{\rm per}-\epsilon_{\rm F}|^{-1}(1-\Delta)}\norm{(1-\Delta)^{-1} V}
\\ &\leq& C \norm{(1-\Delta)^{-1} V} \leq C(\norm{V_1}_{L^6}+\norm{V_1}_{L^2}). 
\end{eqnarray*}
We have used the fact that $\FSm|H^0_{\rm per}-\epsilon_{\rm F}|$ is a bounded operator : $H^0_{\rm per}$ is bounded from below and $\FSm = \oneep (\FSh)$. That $\left| \FSm - \FSl \right| ^{-1} \left( 1 - \Delta\right)$ is also bounded is shown in \cite{CanDelLew-08a}, Lemma~1. For the last step we used the Kato-Seiler-Simon inequality (Theorem~4.1 in \cite{Simon-79})
\begin{equation}\label{eq:KSS}
\left\Vert f(-i\nabla ) g (x) \right\Vert_{\Sch ^p} \leq \left( 2\pi \right) ^{-3/p} \left\Vert f \right\Vert_{L ^p}\left\Vert g\right\Vert_{L ^p}
\end{equation}
for $p\geq2$. We thus obtain
$$|\tr(Q^{--}(1-X_R^2-Y_R^2)V)|\leq CR^{-1}\norm{Q^{--}}_{\gS^1}(\norm{\nabla V_1}_{L^2}+\norm{V_2}_{L^2})$$
as expected. In all these estimates the kinetic energy was not useful. For $Q^{++}$ we have to use it. We start with
\begin{align*}
|\tr(Q^{++}(1-X_R^2-Y_R^2)V_1)|&\leq\norm{|H^0_{\rm per}-\epsilon_{\rm F}|^{1/2}Q^{++}}_{\gS^1}\norm{1-X_R^2-Y_R^2}\norm{V_1|H^0_{\rm per}-\epsilon_{\rm F}|^{-1/2}}\\
&\leq CR^{-1}\norm{|H^0_{\rm per}-\epsilon_{\rm F}|^{1/2}Q^{++}}_{\gS^1}\norm{\nabla V_1}_{L^2}.
\end{align*}
This time we have used that, by Lemma~1 in~\cite{CanDelLew-08a} and~\eqref{eq:KSS} again, 
$$\norm{V_1|H^0_{\rm per}-\epsilon_{\rm F}|^{-1/2}}\leq \norm{V_1(1-\Delta)^{-1/2}}_{\gS^6}\norm{(1-\Delta)^{1/2}|H^0_{\rm per}-\epsilon_{\rm F}|^{-1/2}}\leq C\norm{V_1}_{L^6}\leq C\norm{\nabla V_1}_{L^2}.$$
The last bound is the Sobolev inequality.
For $V_2$ we have to use the full kinetic energy:
\begin{align*}
&\tr(Q^{++}(1-X_R^2-Y_R^2)V_2)\\
&\quad=\tr(|H^0_{\rm per}-\epsilon_{\rm F}|^{1/2}Q^{++}|H^0_{\rm per}-\epsilon_{\rm F}|^{1/2}(1-X_R^2-Y_R^2)|H^0_{\rm per}-\epsilon_{\rm F}|^{-1/2}V_2|H^0_{\rm per}-\epsilon_{\rm F}|^{-1/2})\\
&\qquad-\tr(|H^0_{\rm per}-\epsilon_{\rm F}|^{1/2}Q^{++}|H^0_{\rm per}-\epsilon_{\rm F}|^{1/2}\big[|H^0_{\rm per}-\epsilon_{\rm F}|^{-1/2},X_R^2+Y_R^2\big]V_2|H^0_{\rm per}-\epsilon_{\rm F}|^{-1/2})
\end{align*}
The first term is treated exactly like for $V_1$ whereas for the second term one has to use that 
\begin{multline*}
\norm{\big[|H^0_{\rm per}-\epsilon_{\rm F}|^{-1/2},X_R^2+Y_R^2\big]\,|H^0_{\rm per}-\epsilon_{\rm F}|^{1/2}}\\
=\norm{|H^0_{\rm per}-\epsilon_{\rm F}|^{-1/2}\big[|H^0_{\rm per}-\epsilon_{\rm F}|^{1/2},X_R^2+Y_R^2\big]}=O(R^{-1}) 
\end{multline*}
which is proved as in \cite{CanDelLew-08a}, Lemma~11.

Let us now turn to the double commutators in~\eqref{eq:IMS-V}. We claim that
$$\norm{[X_R,[X_R,V]]}_{\gS^2}\leq CR^{-1}(\norm{\nabla V_1}_{L^2}+\norm{V_2}_{L^2}).$$
To see this we use that $\FSm + \FSmp = 1$ and $[V,\chiR] = 0$ to compute
$$[X_R,V]=\FSm\chi_R[\FSm,V]+[\FSm,V]\chi_R\FSm-(\FSm)^\perp\chi_R[\FSm,V]-[\FSm,V]\chi_R(\FSm)^\perp.$$
We can then write
$$\FSm\chi_R[\FSm,V]=\FSm\chi_R\FSm[\FSm,V]+[\FSm,\chi_R](\FSm)^\perp[\FSm,V].$$
Noting that $\FSm\chi_R\FSm$ commutes with $X_R$, we get
$$[X_R,[X_R,V]]=\FSm\chi_R\FSm\big[X_R,[\FSm,V]\big]+\big[X_R,[\FSm,\chi_R](\FSm)^\perp[\FSm,V]\big]+\text{similar terms}.$$
In the second term of the right side the last commutator is not useful and we can simply bound
$$\norm{\big[X_R,[\FSm,\chi_R](\FSm)^\perp[\FSm,V]\big]}_{\gS^2}\leq 2\norm{X_R}\norm{[\FSm,\chi_R]}\norm{[\FSm,V]}_{\gS^2}\leq CR^{-1}\norm{[\FSm,V]}_{\gS^2}$$
where we have used $\norm{X_R}\leq 1$ and $\norm{\FSmp}\leq 1$. So our last task is to show that
$$\norm{\big[X_R,[\FSm,V]\big]}_{\gS^2}\leq CR^{-1}(\norm{\nabla V_1}_{L^2}+\norm{V_2}_{L^2}).$$
To prove this estimate we express the double commutator as
\begin{equation}
\big[X_R,[\FSm,V]\big]=\FSm\big[[\chi_R,\FSm],V\big](\FSm)^\perp + (\FSm)^\perp\big[[\chi_R,\FSm] ,V\big]\FSm.
\label{eq:horrible}
\end{equation}
To see that \eqref{eq:horrible} holds, note that since $X_R$ commutes with $\FSm$,
$$\FSm\big[X_R,[\FSm,V]\big]\FSm=\big[X_R,\FSm[\FSm,V]\FSm\big]=0$$
because $\FSm[\FSm,V]\FSm=(\FSm)^2V\FSm-\FSm V(\FSm)^2=0$. The argument is the same for $(\FSm)^\perp\big[X_R,[\FSm,V]\big](\FSm)^\perp$. We deduce that the double commutator is purely off-diagonal,
$$\big[X_R,[\FSm,V]\big]=\FSm\big[X_R,[\FSm,V]\big](\FSm)^\perp+(\FSm)^\perp\big[X_R,[\FSm,V]\big]\FSm.$$
Now we compute (using again that $[X_R,\FSm] = 0$ and $\FSm + \FSmp = 1$)
\begin{align*}
\FSm\big[X_R,[\FSm,V]\big](\FSm)^\perp&=\FSm\big[X_R,\FSm[\FSm,V](\FSm)^\perp\big](\FSm)^\perp\\
&= \FSm\big[X_R,\FSm V(\FSm)^\perp\big](\FSm)^\perp\\
&=\FSm\big[X_R,V\big](\FSm)^\perp\\
&=\FSm\chi_R\FSm V(\FSm)^\perp-\FSm V(\FSm)^\perp \chi_R(\FSm)^\perp\\
&=\FSm[\chi_R,\FSm] V(\FSm)^\perp-\FSm V[(\FSm)^\perp, \chi_R](\FSm)^\perp\\
&=\FSm[\chi_R,\FSm] V(\FSm)^\perp-\FSm V[\chi_R,\FSm](\FSm)^\perp\\
&=\FSm\big[[\chi_R,\FSm], V\big](\FSm)^\perp.
\end{align*}
This proves~\eqref{eq:horrible}.

Now $\big[[\chi_R,\FSm],V\big]$ is estimated as usual by expressing $\FSm$ using Cauchy's formula:
$$\FSm=-\frac{1}{2i\pi}\oint_{\mathscr{C}} \frac{dz}{H^0_{\rm per}-z}$$
where $\mathscr{C}$ is a curve enclosing the spectrum of $H^0_{\rm per}$ below $\FSl$.
The formula
\begin{equation}\label{eq:formule utile}
\left[ (z-A) ^{-1},B \right] = (z-A) ^{-1} [A,B] (z-A) ^{-1}
\end{equation}
then leads to (with the standard notation $p=-i\nabla$)
$$[\chi_R,\FSm]=\frac{1}{2\pi}\oint dz\frac{1}{H^0_{\rm per}-z}\big(p\cdot \nabla\chi_R+\nabla\chi_R\cdot p\big)\frac{1}{H^0_{\rm per}-z}.$$
So we get for instance
$$[\chi_R,\FSm]V_2=\frac{1}{2\pi}\oint dz\frac{1}{H^0_{\rm per}-z}\big(2p\cdot \nabla\chi_R+i\Delta \chi_R\big)\frac{1}{H^0_{\rm per}-z}V_2.$$
Using~\eqref{eq:KSS} and the fact that $\|\nabla\chi_R\|_{L^\ii}=O(R^{-1})$ and $\|\Delta\chi_R\|_{L^\ii}=O(R^{-2})$, we easily get
$$\norm{[\chi_R,\FSm] V_2}_{\gS^2}\leq CR^{-1}\norm{V_2}_ {L^2}.$$
We argue the same when $V_2$ is on the left.

For $V_1$ we need the commutator: 
\begin{multline*}
\big[[\chi_R,\FSm],V_1\big]=\frac{1}{2\pi}\oint dz\frac{1}{H^0_{\rm per}-z}\big(p\cdot \nabla\chi_R+\nabla\chi_R\cdot p\big)\left[\frac{1}{H^0_{\rm per}-z},V_1\right]\\-\frac{i}{\pi}\oint dz\frac{1}{H^0_{\rm per}-z}\nabla\chi_R\cdot \nabla V_1\frac{1}{H^0_{\rm per}-z}\\+\frac{1}{2\pi}\oint dz\left[\frac{1}{H^0_{\rm per}-z},V_1\right]\big(p\cdot \nabla\chi_R+\nabla\chi_R\cdot p\big)\frac{1}{H^0_{\rm per}-z}
\end{multline*}
and we argue as before. For the commutator on the last line we use \eqref{eq:formule utile} to write 
$$ \left[\frac{1}{H^0_{\rm per}-z},V_1\right] = \frac{1}{H^0_{\rm per}-z} \left[-\Delta, V_1 \right] \frac{1}{H^0_{\rm per}-z}$$
and follow arguments from \cite{CanDelLew-08a}, Lemma~5.

All in all, we have shown that for $V=V_1+V_2$
$$|\tr\big(Q(V-X_RVX_R-Y_RVY_R)\big)|\leq CR^{-1}\norm{Q}_{\cQ}\big(\norm{\nabla V_1}_{L^2}+\norm{V_2}_{L^2}\big)$$
which, by duality, precisely proves \eqref{eq:error_localization_rho}.

\subsubsection*{Proof of~\eqref{eq:localization_kinetic}: localization of the kinetic energy}
We first remark that
$$(\chi_R^2)^{--}=\FSm\chi_R\left(\FSm+(\FSm)^\perp\right)\chi_R\FSm=(X_R^2)^{--}+[\FSm,\chi_R](\FSm)^\perp[\chi_R,\FSm]$$
and a similar equality for $(\chi_R^2)^{++}$. Since by construction 
$$X_R ^2=\left( X_R ^2\right) ^{++}+\left( X_R ^2\right) ^{--}$$
this yields
$$(\chi_R^2)^{--}+(\chi_R^2)^{++}=X_R^2-[\FSm,\chi_R]^2\qquad\text{and}\qquad (\eta_R^2)^{--}+(\eta_R^2)^{++}=Y_R^2-[\FSm,\eta_R]^2.$$
From this we deduce that
\begin{align*}
\tr_0(H^0_{\rm per}-\epsilon_{\rm F})Q&=\tr(H^0_{\rm per}-\epsilon_{\rm F})(Q^{++}+Q^{--})\\
&=\tr\frac{(\chi_R^2+\eta_R^2)(H^0_{\rm per}-\epsilon_{\rm F})+(H^0_{\rm per}-\epsilon_{\rm F})(\chi_R^2+\eta_R^2)}{2}(Q^{++}+Q^{--})\\
&=\tr\frac{(X_R^2+Y_R^2)(H^0_{\rm per}-\epsilon_{\rm F})+(H^0_{\rm per}-\epsilon_{\rm F})(X_R^2+Y_R^2)}{2}(Q^{++}+Q^{--})\\
&\qquad -\tr\frac{[\FSm,\chi_R]^2(H^0_{\rm per}-\epsilon_{\rm F})+(H^0_{\rm per}-\epsilon_{\rm F})[\FSm,\chi_R]^2}{2}(Q^{++}+Q^{--})\\
&\qquad -\tr\frac{[\FSm,\eta_R]^2(H^0_{\rm per}-\epsilon_{\rm F})+(H^0_{\rm per}-\epsilon_{\rm F})[\FSm,\eta_R]^2}{2}(Q^{++}+Q^{--})
\end{align*}
hence that
\begin{align*}
&\tr_0(H^0_{\rm per}-\epsilon_{\rm F})Q-\tr_0(H^0_{\rm per}-\epsilon_{\rm F})X_RQX_R-\tr_0(H^0_{\rm per}-\epsilon_{\rm F})Y_RQY_R\\
&\qquad\quad=\frac12\tr\big([X_R,[X_R,H^0_{\rm per}]]+[Y_R,[Y_R,H^0_{\rm per}]]\big)(Q^{++}+Q^{--})\\
&\qquad\qquad -\frac12\tr\big([\FSm,\chi_R]^2(H^0_{\rm per}-\epsilon_{\rm F})+(H^0_{\rm per}-\epsilon_{\rm F})[\FSm,\chi_R]^2\big)(Q^{++}+Q^{--})\\
&\qquad\qquad -\frac12\tr\big([\FSm,\eta_R]^2(H^0_{\rm per}-\epsilon_{\rm F})+(H^0_{\rm per}-\epsilon_{\rm F})[\FSm,\eta_R]^2\big)(Q^{++}+Q^{--}).
\end{align*}
We conclude that
\begin{multline*}
\left|\tr_0(H^0_{\rm per}-\epsilon_{\rm F})Q-\tr_0(H^0_{\rm per}-\epsilon_{\rm F})X_RQX_R-\tr_0(H^0_{\rm per}-\epsilon_{\rm F})Y_RQY_R\right|\\
\leq C\norm{Q}_{\cQ}\bigg(\norm{|H^0_{\rm per}-\epsilon_{\rm F}|^{-1/2}[X_R,[X_R,H^0_{\rm per}]]|H^0_{\rm per}-\epsilon_{\rm F}|^{-1/2}}\\
+\norm{|H^0_{\rm per}-\epsilon_{\rm F}|^{-1/2}[Y_R,[Y_R,H^0_{\rm per}]]|H^0_{\rm per}-\epsilon_{\rm F}|^{-1/2}}\\
+\norm{[\FSm,\chi_R]^2|H^0_{\rm per}-\epsilon_{\rm F}|^{1/2}}+\norm{[\FSm,\eta_R]^2|H^0_{\rm per}-\epsilon_{\rm F}|^{1/2}}\bigg).
\end{multline*}
For the last term we recall from \cite{CanDelLew-08a}, Lemma~10, that 
$$ \norm{ [\FSm,\etaR]} \leq C R ^{-1}$$
and note that the same proof can be employed to show that 
$$ \norm{ [\FSm,\etaR] |\FSh-\FSl| ^{1/2} } \leq C R ^{-1}.$$
The second to last term is treated similarly.
For the double commutators, a computation shows that
\begin{multline*}
[X_R,[X_R,H^0_{\rm per}]]=(\FSm)^\perp\Big([\chi_R,\FSm]\,[\chi_R,\Delta]+[\chi_R,\Delta]\,[\chi_R,\FSm]+|\nabla\chi_R|^2\Big)(\FSm)^\perp\\
-\FSm\Big([\chi_R,\FSm]\,[\chi_R,\Delta]+[\chi_R,\Delta]\,[\chi_R,\FSm]+|\nabla\chi_R|^2\Big)\FSm.
\end{multline*}
We have $[\chi_R,\Delta]=(\Delta\chi_R)+2i\nabla\chi\cdot p$ with $p=-i\nabla$. Using then that $p|H^0_{\rm per}-\epsilon_{\rm F}|^{-1/2}$ is bounded and the fact that $\norm{[\chi_R,\FSm]}=O(R^{-1})$, we conclude similarly as before that
$$\norm{|H^0_{\rm per}-\epsilon_{\rm F}|^{-1/2}[X_R,[X_R,H^0_{\rm per}]]|H^0_{\rm per}-\epsilon_{\rm F}|^{-1/2}}=O\left(\frac{1}{R^2}\right).$$
The term involving $Y_R$ is treated similarly. This ends the proof of Lemma~\ref{lem:loc_properties}.\hfill\qed

\bibliographystyle{siam}
\bibliography{biblio}

\end{document}